\documentclass[journal]{IEEEtran}

\usepackage{cite}
\usepackage{amsmath,amssymb,amsfonts,amsthm}
\usepackage{graphicx}
\usepackage{textcomp}
\usepackage{xcolor}
\usepackage{cases}
\usepackage[utf8]{inputenc}
    
\usepackage[acronym,toc,shortcuts]{glossaries}
\newacronym{SA}{SA}{Simulated Annealing}
\newacronym{VLC}{VLC}{visible light communications}
\newacronym{RF}{RF}{radio frequency}
\newacronym{V2V}{V2V}{vehicle-to-vehicle}
\newacronym{V2X}{V2X}{vehicle-to-everything}
\newacronym{B5G}{B5G}{beyond-fifth generation}
\newacronym{5G}{5G}{fifth-generation}
\newacronym{6G}{6G}{sixth-generation}
\newacronym{LED}{LED}{light emitting diode}
\newacronym{LEDs}{LEDs}{light emitting diodes}
\newacronym{OMA}{OMA}{orthogonal multiple access}
\newacronym{FDMA}{FDMA}{frequency-division multiple-access}
\newacronym{TDMA}{TDMA}{time-division multiple-access}
\newacronym{CDMA}{CDMA}{code-division multiple-access}
\newacronym{OFDMA}{OFDMA}{orthogonal frequency-division multiple-access}
\newacronym{WDMA}{WDMA}{wavelength-division multiple-access}
\newacronym{NOMA}{NOMA}{non-orthogonal multiple access}
\newacronym{PD-NOMA}{PD-NOMA}{power-domain NOMA}
\newacronym{SC}{SC}{superposition coding}
\newacronym{SIC}{SIC}{successive interference cancellation}
\newacronym{BS}{BS}{base station}
\newacronym{QoS}{QoS}{quality-of-service}
\newacronym{NP}{NP}{non-deterministic polynomial-time}
\newacronym{DCO-OFDM}{DCO-OFDM}{direct-current biased optical-OFDM}
\newacronym{ITU}{ITU}{international telecommunication union}
\newacronym{FoV}{FoV}{field-of-view}
\newacronym{FoVs}{FoVs}{field-of-views}
\newacronym{CSI}{CSI}{channel state information}
\newacronym{LACO-OFDM}{LACO-OFDM}{layered asymmetrically clipped optical OFDM}
\newacronym{FR}{FR}{frequency reuse}
\newacronym{EAs}{EAs}{evolutionary algorithms}
\newacronym{C-LiAN}{C-LiAN}{centralized light access network}
\newacronym{AP}{AP}{access point}
\newacronym{APs}{APs}{access points}
\newacronym{PD}{PD}{photo-diode}
 \newacronym{PDs}{PDs}{photo-diodes}
\newacronym{SINR}{SINR}{signal-to-noise-interference ratio}
\newacronym{LoS}{LoS}{line-of-sight}
\newacronym{AWGN}{AWGN}{additive white Gaussian noise}
\newacronym{SNR}{SNR}{signal-to-noise ratio}
\newacronym{NLUPA}{NLUPA}{next-largest-difference user-pairing algorithm}
\newacronym{D-NLUPA}{D-NLUPA}{divide-and-next-largest-difference user-pairing algorithm}
\newacronym{NAICS}{NAICS}{network-assisted interference cancellation and suppression}
\newacronym{LTE}{LTE}{long term evolution}
\newacronym{3GPP}{3GPP}{3rd Generation Partnership Project}
\newacronym{CR}{CR}{cognitive radio}
\newacronym{2D}{2D}{two-dimension}
\newacronym{GP}{GP}{gradient projection}
\newacronym{umMTC}{umMTC}{ultra-massive machine-type communication}
\newacronym{IoE}{IoE}{internet-of-everything}
\newacronym{ZF}{ZF}{zero-forcing}
\newacronym{NLIP}{NLIP}{non-linear integer programming}
\newacronym{DP}{DP}{dynamic programming}
\newacronym{THz}{THz}{Terahertz}
\newacronym{D2D}{D2D}{device-to-device}
\newacronym{mmWave}{mmWave}{Millimeter-wave}
\newacronym{PS}{PS}{phase shifter}
\newacronym{PA}{PA}{power amplifier}
\newacronym{ABF}{ABF}{analog beamforming}
\newacronym{SPS}{SPS}{single-phase shifter}
\newacronym{AWV}{AWV}{antenna weight vector}
\newacronym{EE}{EE}{energy efficiency}
\newacronym{CEE}{CEE}{computation energy efficiency}
\newacronym{SE}{SE}{spectral efficiency}
\newacronym{NLoS}{NLoS}{non-line-of-sight}
\newacronym{mMIMO}{mMIMO}{massive multiple-input multiple-output}
\newacronym{MISO}{MISO}{multiple-input single-output}
\newacronym{SIMO}{SIMO}{single-input multiple-output}
\newacronym{SISO}{SISO}{single-input single-output}
\newacronym{MOO}{MOO}{multi-objective optimization}
\newacronym{HP}{HP}{hybrid beamforming}
\newacronym{SCA}{SCA}{successive convex approximation }
\newacronym{SOC}{SOC}{second order cone}
\newacronym{RHS}{RHS}{right-hand side}
\newacronym{SDR}{SDR}{semidefinite relaxation}
\newacronym{VR}{VR}{virtual reality}
\newacronym{AR}{AR}{augmented reality}
\newacronym{MAR}{MAR}{mobile augmented reality}
\newacronym{NR}{NR}{new radio}
\newacronym{FD}{FD}{full-duplex}
\newacronym{HD}{HD}{half-duplex}
\newacronym{QoE}{QoE}{quality-of-experience}
\newacronym{ULA}{ULA}{uniform linear array}
\newacronym{SI}{SI}{self-interference}
\newacronym{DF}{DF}{decode-and-forward}
\newacronym{GEE}{GEE}{global energy efficiency}
\newacronym{CNOMA}{CNOMA}{cooperative NOMA}
\newacronym{MEC}{MEC}{mobile edge computing}
\newacronym{IoT}{IoT}{internet-of-things}
\newacronym{CPU}{CPU}{central processing unit}
\newacronym{KKT}{KKT}{Karush–Kuhn–Tucker}
\newacronym{ML}{ML}{machine learning}
\newacronym{CM}{CM}{constant modulus}

\usepackage{epstopdf}

\usepackage{cuted}
\usepackage{needspace}  
\newenvironment{bottomborder}%
{}  
{\needspace{\baselineskip}\hrule} 

\usepackage{hyperref}
\hypersetup{colorlinks=true,linkcolor=blue,urlcolor=black,}

\usepackage{dirtytalk}

\usepackage{enumitem}

\usepackage{algpseudocode}
\usepackage[linesnumbered,ruled,vlined]{algorithm2e}

\usepackage{ragged2e}

\usepackage{multirow}
\usepackage{dcolumn,booktabs}

\usepackage{subcaption}

\newtheorem{proposition}{Proposition}

\begin{document}

\title{Energy-Efficient Optimization of Multi-User NOMA-Assisted Cooperative THz-SIMO MEC Systems}

\author{Omar~Maraqa, Saad Al-Ahmadi, Aditya Rajasekaran, Hamza Sokun, Halim Yanikomeroglu, \IEEEmembership{Fellow,~IEEE}, Sadiq M. Sait, \IEEEmembership{Senior~Member,~IEEE}

\thanks{O.~Maraqa and S. Al-Ahmadi are with the Department of Electrical Engineering, S.~M.~Sait is with the Department of Computer Engineering, King Fahd University of Petroleum \& Minerals, Dhahran-31261, Saudi Arabia (e-mails: \{g201307310; saadbd; sadiq\}@kfupm.edu.sa). A.~S.~Rajasekaran and H.~U.~Sokun is with Ericsson Canada Inc, Ottawa, ON K2K 2V6, Canada (emails: \{aditya.sriram.rajasekaran; hamza.sokun\}@ericsson.com). H. Yanikomeroglu is with the Department of Systems and Computer Engineering, Carleton University, Ottawa, ON K1S 5B6, Canada (email: halim@sce.carleton.ca). This work was supported by the interdisciplinary research center for communication systems and sensing (IRC-CSS), King Fahd University of Petroleum and Minerals, under Grant number INCS2107. Part of this work were presented at IEEE PIMRC 2021~\cite{maraqa2021energy}.}%
\thanks{$\copyright$ 2023 IEEE. Personal use of this material is permitted. Permission from IEEE must be obtained for all other uses, in any current or future media, including reprinting/republishing this material for advertising or promotional purposes, creating new collective works, for resale or redistribution to servers or lists, or reuse of any copyrighted component of this work in other works.}%

}

\markboth{Accepted in IEEE TRANSACTIONS ON COMMUNICATIONS, 2023}
{Maraqa \MakeLowercase{\textit{et al.}}: Energy-Efficient Optimization of Multi-User NOMA-Assisted Cooperative THz-SIMO MEC Systems}

\maketitle

\begin{abstract}
The various requirements in terms of data rates and latency in beyond 5G and 6G networks have motivated the integration of a variety of communications schemes and technologies to meet these requirements in such networks. Among these schemes are Terahertz (THz) communications, cooperative non-orthogonal multiple-access (NOMA)-enabled schemes, and mobile edge computing (MEC). THz communications offer abundant bandwidth for high-data-rate short-distance applications and NOMA-enabled schemes are promising schemes to realize the target spectral efficiencies and low latency requirements in future networks, while MEC would allow distributed processing and data offloading for the emerging applications in these networks. In this paper, an energy-efficient scheme of multi-user NOMA-assisted cooperative THz single-input multiple-output (SIMO) MEC systems is proposed to allow the uplink transmission of offloaded data from the far cell-edge users to the more computing resources in the base station (BS) through the cell-center users. To reinforce the performance of the proposed scheme, two optimization problems are formulated and solved, namely, the first problem minimizes the total users' energy consumption while the second problem maximizes the total users' computation energy efficiency (CEE) for the proposed scheme. In both problems, the NOMA user pairing, the BS receive beamforming, the transmission time allocation, and the NOMA transmission power allocation coefficients are optimized, while taking into account the full-offloading requirements of each user as well as the predefined latency constraint of the system. The obtained results reveal new insights into the performance and design of multi-user NOMA-assisted cooperative THz-SIMO MEC systems. Particularly, with relatively high offloading rate demands (several Gbits/user), we show that (i) the proposed scheme can handle such demands while satisfying the predefined latency constraint, and (ii) the full-offloading model can be considered the most effective solution in conserving mobile devices' resources as compared to the system with the partial-offloading model or the system without offloading. 
\end{abstract}


\section{Introduction}
\label{Sec:Introduction}
\IEEEPARstart{T}{}he emerging high-data-rate and ultra-low latency applications in \ac{B5G} and \ac{6G} networks such as \ac{MAR} applications~\cite{9363323} require larger corresponding bandwidths with more strict latency constraints than the ones offered by \ac{5G} networks~\cite{7980118}. This has spurred the interest in the \ac{THz} bands to realize the target data rates. Moreover, the large amount of data to be processed in such networks has motivated \ac{MEC}, also known as Multi-Access Edge Computing, for remote computation of intensive-tasks limited-battery mobile devices. In \ac{MEC} networks, the \ac{BS} is equipped with \ac{MEC} server(s) that can receive and execute offloaded tasks from the network mobile devices. Later, the computation results are sent to these mobile devices after execution~\cite{9622148}. Power-domain \ac{NOMA} scheme, on the other hand, allow multiple users to share the same resource (e.g., a time/frequency resource block) and separate the users in the power domain with some additional receiver complexity~\cite{9154358}. Hence, combining these novel network technologies, (i.e., THz, NOMA, and MEC) is of interest to realize the unprecedented demands for the users in future wireless networks while efficiently utilizing the costly wireless spectrum resources in an energy-efficient manner.

There are three computing offloading models in \ac{MEC} networks, namely, binary, partial, and full offloading~\cite{7879258}. In the binary-offloading model, the user tasks can be either locally computed at the mobile device or remotely computed at the \ac{MEC} server and cannot be partitioned. In the partial-offloading model, the user tasks are divided into two parts: the local computing part and the offloading part. While in the full-offloading model, the whole user tasks are offloaded and remotely accomplished by the \ac{MEC} server(s)~\cite{9154358}. In the era of \ac{B5G} networks, a wide range of mobile-initiated applications that require complex computation such as \ac{ML} and signal processing algorithms are on the rise. These applications impose a heavy computation burden on the limited-battery user devices. With the full-offloading model, a user device can be responsible only for collecting the input data and displaying the computation results. Hence, in this paper, the full-offloading model is adopted. In addition, the obtained results in Section~\ref{Section: Results and discussions} quantitatively support that adopting the full-offloading model as the most effective solution in conserving mobile devices' resources, especially for the systems that require high task input-bits to be processed.

The amount of work in user-assisted cooperative \ac{MEC} \ac{NOMA} networks is still limited when it comes to the adopted system model~\cite{8951269,8823868,9417585,9613252,9348649,9325063,9566305,9417425}. Specifically, all these works have analyzed a cooperative uplink \ac{SISO} system and they are different at where the \ac{NOMA} scheme is adopted (i.e., at the far-user side~\cite{8951269,8823868,9417585,9613252}, or at the near-user side~\cite{9348649}, or at both far-user and near-user sides~\cite{9325063,9566305}, or at both \ac{BS} and near-user sides~\cite{9417425}). In particular, in~\cite{9348649}, a three-node system that comprises of a far-user, a near-user, and a \ac{BS} has been analyzed to minimize the users' sum-energy consumption. The authors solved an optimization problem that involves a joint transmission time assignment and power allocation for users. Through numerical simulations, the authors demonstrated the superiority of the proposed scheme compared to (i) the \ac{THz}-\ac{OMA} counterpart, (ii) the direct transmission scheme, and (iii) an equivalent system with random time and power allocation. In this work, motivated by~\cite{9348649}, we adopt the \ac{NOMA} scheme at the cell-center users, but on the contrary to~\cite{8951269,8823868,9417585,9613252,9348649,9417425,9325063,9566305}, (i) the work here focuses on the cooperative uplink \ac{SIMO} system model, (ii) the adopted band in the previously mentioned works was sub-$6$ GHz band compared to our work here that embraces the \ac{THz} band, and (iii) our work is based on a multi-user system model, in contrast, all the works in~\cite{8951269,9348649,9417425,8823868,9613252,9325063} have adopted three-node systems, except for~\cite{9417585,9566305} that adopted a device-assisted multi-helper \ac{MEC} system which differs from our proposed multi-user model. Recently, a multi-antenna \ac{NOMA}-assisted wireless powered \ac{MEC} with user cooperation has been proposed in~\cite{9345931} but still for a two-user system model and sub-$6$ GHz operation band. 

In this paper, an energy-efficient scheme of \ac{NOMA}-assisted cooperative \ac{THz}-\ac{SIMO} \ac{MEC} systems is proposed to allow the uplink transmission of offloaded data from the far cell-edge users to the \ac{BS}. The proposed cooperative scheme comprises four stages: (i) a user pairing stage that forms \ac{NOMA} user pairs, where a cell-edge user is paired with a cell-center user, (ii) a receive beamforming stage in the \ac{BS} to schedule each \ac{NOMA} user-pair in the system, (iii) a time allocation stage to schedule and send the offloading task input-bits from the formulated \ac{NOMA} user pairs to the \ac{BS} within a predefined latency constraint, and (iv) a \ac{NOMA} transmission power allocation stage. The performance of the proposed scheme is reinforced by solving two energy-related optimization problems. Thus, our contributions can  be summarized as follows:
\begin{itemize}
    \item Proposing an energy-efficient \ac{NOMA}-assisted cooperative \ac{THz}-\ac{SIMO} \ac{MEC} scheme for allowing the far cell-edge users to transmit their offloading data to the \ac{BS} through cooperating cell-center users.
    \item Formulating an energy consumption minimization problem for all users that consists of a joint design of user pairing, beamforming, time, and power allocation. Such a problem is non-convex and intractable to solve jointly. Hence, a low complexity solution is adopted to solve this problem, where (i) the \ac{NOMA} user pairing is solved by the one-to-one matching Hungarian algorithm, (ii) the \ac{BS} receive beamforming is optimized using analog beamforming with the aid of the cosine similarity metric, and (iii) the time and power allocations are jointly optimized and tackled using the Lagrange duality method.
    \item Formulating a \ac{CEE} maximization problem for all users that consists of a joint design of user pairing, time allocation, beam gain and power allocation, and beamforming. Similar to the previous optimization problem, this problem is non-convex and a low-complexity solution is also adopted to solve this problem. Toward this solution, (i) the \ac{NOMA} user pairing is solved using the Hungarian algorithm, (ii) the transmission time for the data offloading in the links between the cell-edge user and the cell-center user and the link between the cell-center user and the BS is divided equally, and (iii) the power allocation as well as the \ac{BS} receive beamforming vectors are optimized using some mathematical relaxation procedures with the aid of the Dinkelbach algorithm.
    \item Illustrating through simulations that the proposed scheme can handle relatively high offloading rate demands (several Gbits/user) within a predefined latency constraint and outperforms several baseline schemes. Specifically, our proposed scheme, (i) consumes much less total users' energy compared to the partial-offloading model and without offloading model, (ii) handles several Gbits of offloading data for each user compared to tens of Mbits of offloading data in the mmWave counterpart system, (iii) consumes slightly less total users' energy compared to the THz-OMA counterpart with the advantage of being more resource-efficient, (iv) provides a higher total users' \ac{CEE} compared to the systems with partial offloading and without offloading, and (v) achieves a significant increase in the total users' \ac{CEE} compared to its \ac{THz}-\ac{OMA} counterpart.
\end{itemize}

The rest of this paper is organized as follows. In Section~\ref{Section: System Model and Problem Formulation}, the \ac{THz} channel model and the offloading model of the proposed \ac{NOMA}-assisted cooperative \ac{THz}-\ac{SIMO} \ac{MEC} scheme are presented. In Section~\ref{Sec: The Proposed NOMA-assisted Cooperative THz-SIMO MEC Scheme}, the formulation and solution of the total users' energy consumption minimization problem are presented. In Section~\ref{Section: Total Users' Computation Energy Efficiency Maximization}, the formulation and solution of the total users' \ac{CEE} maximization problem are presented. Detailed Simulation results are provided in Section~\ref{Section: Results and discussions}, which is followed by the paper conclusions in Section~\ref{Section: Conclusions}. 
\section{System and Channel Models}
\label{Section: System Model and Problem Formulation}
An uplink communication scheme where a \ac{BS}, that is equipped with $N$ antennas, receives and executes offloading data from a set of $K$ single-antenna \ac{NOMA} user pairs is depicted in Fig.~\ref{fig: THz NOMA system model}. The users are classified into $K_\textnormal{Coop}$ cell-center users and $K_\textnormal{Edge}$ cell-edge users. The links between the \ac{BS} and the cell-edge users are assumed to be very weak due to the high attenuation in \ac{THz} bands, which is the case even for the mmWave band~\cite{9814839}. First, each cell-edge user is paired with a cooperating cell-center user using the Hungarian algorithm, as discussed in Section~\ref{Sec: The Proposed NOMA-assisted Cooperative THz-SIMO MEC Scheme}, to form a \ac{NOMA} user-pair (e.g., $j$-th cell-edge user and $i$-th cooperating cell-center user). In each \ac{NOMA} user-pair, first, the cell-edge user transmits its offloading data to the cell-center user on an orthogonal channel transmission. Then, the cell-center user performs \ac{HD} \ac{DF} \ac{NOMA} cooperation to relay the cell-edge user's offloading data alongside his own offloading data to the \ac{BS} for processing. From the \ac{BS} side, each \ac{NOMA} user-pair is served by an orthogonal resource block. Specifically, each \ac{NOMA} user-pair is served by one receive precoding vector on an orthogonal resource block as \ac{ABF} is adopted~\cite{maraqa2021energy}. The remaining \ac{NOMA} user pairs can be analyzed similarly.

\begin{figure}[t!]
\centering
\includegraphics[width=0.485\textwidth]{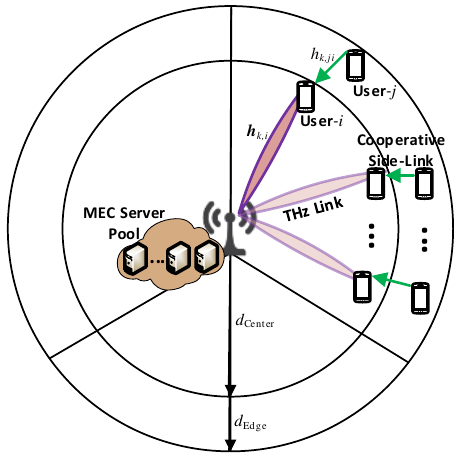}
\caption{The proposed model for multi-user \ac{NOMA}-assisted cooperative \ac{THz}-\ac{SIMO} \ac{MEC} systems.}
\label{fig: THz NOMA system model}
\end{figure}

\begin{figure*}[t!]
\centering
\includegraphics[width=0.9\textwidth]{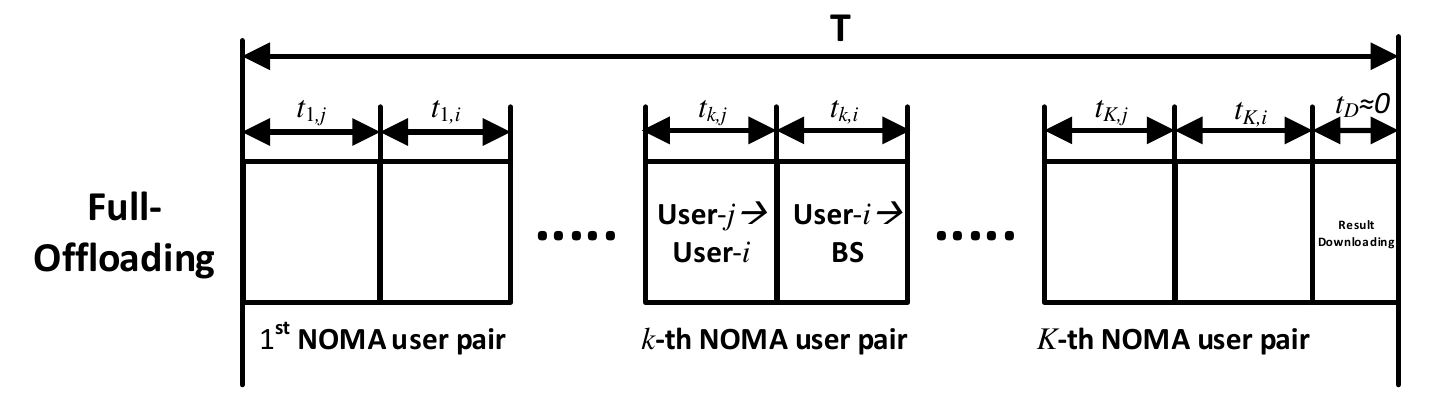}
\caption{Full-offloading model for the proposed multi-user \ac{NOMA}-assisted cooperative \ac{THz}-\ac{SIMO} \ac{MEC} system.}
\label{fig: Full offloading model}
\end{figure*}

As explained in the introduction section, we adopt the full-offloading model to meet the high demands in next-generation networks. For the operation of this model, it is assumed that there is a central controller in the \ac{BS} that collects computation-related information, required \ac{CSI}, and distributes the computation-related policies for all users~\cite{8488502}. At this point, we focus the analysis on the $k$-th \ac{NOMA} user-pair that contains the $j$-th cell-edge user and $i$-th cooperating cell-center user. In this \ac{NOMA} user-pair, and similarly on all other \ac{NOMA} user pairs, we assume that the $j$-th user has a data size of $L_{k,j}>0$ (in bits) computation tasks to execute. Likewise, the $i$-th user has a data size of $L_{k,i}>0$ (in bits) computation tasks to execute. In this \ac{NOMA} user-pair, we denote the time needed for the data to be sent from the $j$-th user to the $i$-th user as $t_{k,j}$, and the time needed for the data to be sent from the $i$-th user to the \ac{BS} as $t_{k,i}$. At the \ac{BS}, the time consumed for the task computation and result downloading can be ignored because of the powerful computation at the \ac{MEC} server pool and the small data size of the computation result~\cite{8951269, 9348649}. Considering latency-critical applications, all offloaded data need to be executed within one time block with duration $T>0$, (i.e., $\sum_{k=1}^{K} (t_{k,j}+t_{k,i}) \leq T$). In Fig.~\ref{fig: Full offloading model}, an illustration of the full-offloading model is provided to show how the block duration $T$ is divided between all the \ac{NOMA} user pairs. Moreover, we assume that the channel coherence time is larger than the block duration $T$, so that the channel gain remains constant during $T$~\cite{8488502}. Next, we discuss the \ac{THz} channel model and the details of offloading model for the considered multi-user \ac{NOMA}-assisted cooperative \ac{THz}-\ac{SIMO} \ac{MEC} system.
\subsection{THz Channel Model}
In the \ac{THz} band, transmitted signals suffer from a high path loss due to the existence of the spreading and molecular absorption losses. This makes \ac{THz} communication links sensitive to the \ac{LoS} obstacles blockage such as walls. For example, it was reported in~\cite{6574880} that when the \ac{LoS} exists between a transmitter and a receiver, the first-order reflections are attenuated on average by more than $10$ dB, and the second-order reflections are attenuated on average by more than $20$ dB. Therefore, in this paper, similar to~\cite{9115278}, we assume the existence of the \ac{LoS} path only. Subsequently, the channel gain receive vector between the $i$-th user in the $k$-th \ac{NOMA} user-pair and the $N$-antenna array at the \ac{BS} can be expressed as 
\begin{equation} \label{eq: THz channel}
\begin{gathered}
\mathbf{h}_{k,i}=  \sqrt{N} ( \sqrt{\frac{G_{t,i} G_{r,\textnormal{BS}}}{\mathcal{PL}(f,d_i)}} \boldsymbol{a} ( \theta_{k,i}) ), \\ \mathcal{PL}(f,d_i) = \mathcal{L}_{\textnormal{spread}}(f,d_i) \mathcal{L}_{\textnormal{abs}}(f,d_i)=  \big( \frac{4 \pi f d_i}{c} \big)^2 e^{k_{\textnormal{abs}}(f)d_i},
\end{gathered}
\end{equation}
\noindent where $\mathcal{PL}(f,d_i)$ denotes the path loss incurred by the \ac{THz} signal that is transmitted from the $i$-th user to the \ac{BS} on frequency $f$ with $d_i$ meters distance separation between them. $\boldsymbol{a} (\theta_{k,i})$ denotes the array steering vector towards the $i$-th user of the $k$-th user-pair. $G_{t,i}$ and $G_{r,\textnormal{BS}}$ respectively denote the $i$-th user antenna transmission gain and the \ac{BS} antenna reception gain. In this paper, we assume that all users have the same antenna transmission and reception gains, and the users and the \ac{BS} are equipped with directive antennas to countermeasure the channel attenuation in the \ac{THz} links~\cite{8610080}. $c$ denotes the speed of light. The path loss of \ac{THz} signals consists of two multiplicative factors, namely, (i) the spreading loss $\mathcal{L}_{\textnormal{spread}}$ and (ii) the molecular absorption loss $\mathcal{L}_{\textnormal{abs}}$. The term $k_{\textnormal{abs}}(f)$ denotes the frequency-dependent absorption coefficient for different isotopologues of water vapor molecules. In a typical medium, water vapor molecules contribute to most of the total absorption loss as compared to the little contributions of the other air molecules~\cite{9115278}. The array steering vector for \ac{ULA} can be expressed as~\cite{9115278}
\begin{equation} \label{eq: steering vector}
\boldsymbol{a} (\theta_{k,i}) =\frac{1}{\sqrt{N}} [1, ..,  e^{j \pi [n \sin(\theta_{k,i})]}, .., e^{j \pi [(N-1) \sin(\theta_{k,i})]} ]^{T}, 
\end{equation}
\noindent where $\theta_{k,i}$ denotes the physical angle-of-arrival of the \ac{THz} beam. The channel model of cooperative side-links, between cell-edge users and cell-center users, also follows the one in~\eqref{eq: THz channel}; since the in-band user-assisted cooperative communications paradigm is adopted~\cite{liu2015band}. Hence, the cooperative side-link channel gain between the $j$-th user and the $i$-th user can be expressed as~\cite{balanis2011modern}
\begin{equation} \label{eq: channel gain between user-i and user-j}
h_{k,ji}=\sqrt{\frac{G_{t,j} G_{r,i}}{ \big( \frac{4 \pi f d_{ji}}{c} \big)^2 e^{k_{\textnormal{abs}}(f)d_{ji}} }}.
\end{equation}

\subsection{Offloading System Model}
\label{Subsection: Offloading Model}

In this subsection, we focus on the offloading data delivery from users to the \ac{BS} that involves two phases, (i) an orthogonal channel transmission from each cell-edge user to each cell-center user, and (ii) a \ac{NOMA} transmission from each cell-center user to the \ac{BS}. Also here, let us focus the analysis on the $k$-th \ac{NOMA} user-pair and all other \ac{NOMA} user pairs can be analyzed in a similar manner; as, in this paper, the analog beamforming at the \ac{BS} is assumed since it is a more energy-saving option compared to digital and hybrid beamforming~\cite{8706964}. During the first phase, the received signal at the $i$-th user, that is transmitted from the $j$-th user, can be written as 
\begin{equation} \label{eq: Rx to user-i in cooperative}
 y_{k,i}^j= \underbrace{h_{k,ji}\sqrt{p_{k,j}} \hat{s}_{k,j}}_{\textnormal{Data from User-}j}+ \underbrace{n_{k,i}}_{\textnormal{Noise}},
\end{equation}
\noindent where $h_{k,ji}$ denotes again the cooperative side-link channel gain between $j$-th user and $i$-th user, $p_{k,j}$ represents the transmit power of the $j$-th user, $\hat{s}_{k,j}$ represents the transmit offloading data from the $j$-th user, and $n_{k,i}$ denotes the noise seen at the $i$-th user with a power $\sigma^2$. Consequently, the achievable offloading data rate of the $j$-th user at the $i$-th user can be expressed as 
\begin{equation} \label{eq:  rate at user-i}
  R_{k,i}^{j}=W \log_2 (1 + \frac{p_{k,j} |h_{k,ji}|^2}{\sigma^2}),
\end{equation}
\noindent where $W$ is the available contiguous bandwidth of the adopted \ac{THz} channel window. In this \ac{NOMA} user-pair, the duration of the first phase is $t_{k,j}$. Subsequently, the energy consumption for data offloading of the $j$-th user is $E_{k,j}^{\textnormal{off}}=t_{k,j} p_{k,j}$. 

During the second phase, the received superimposed signal at the \ac{BS} from the $i$-th cooperating cell-center user of the $k$-th user-pair is 
\begin{equation} \label{eq: Rx at BS}
\begin{gathered}
y_{k,\textnormal{BS}} = \underbrace{\mathbf{h}_{k,i}^{H} \mathbf{w}_{k} \sqrt{\beta_{k,i}p_{k,i}} s_{k,i}}_{\textnormal{Data from User-}i} + \underbrace{\mathbf{h}_{k,i}^{H} \mathbf{w}_{k}  \sqrt{\beta_{k,j} p_{k,i}} s_{k,j}}_{\textnormal{Data from User-}j} \\ \qquad \qquad \qquad \ \ + \underbrace{\mathbf{n}_{k}^{H}\mathbf{w}_{k}}_{\textnormal{Noise}},
\end{gathered}
\end{equation}
\noindent where $\mathbf{w}_k$ represents the receive vector at the \ac{BS} antenna array, $\beta_{k,i}$ and $\beta_{k,j}$ denote the \ac{NOMA} power fractions allocated to transmit the $i$-th cooperating cell-center user and $j$-th cell-edge user offloaded data bits, respectively, in which $\beta_{k,j}~+~\beta_{k,i}$~=~$1$~\cite{9348649}. $s_{k,i}$ and $s_{k,j}$ denote the transmit offloaded data messages of the $i$-th user and $j$-th user. $p_{k,i}$ denotes the transmit power of the $i$-th user. $\mathbf{n}_{k}$ denotes the $N$-dimensional Gaussian white noise vector seen at the \ac{BS} with a power $\sigma^2$. The second term in~\eqref{eq: Rx at BS}, represents the offloading data of the $j$-th user that was embedded with the $i$-th user offloading data through superposition coding (SC). In this \ac{NOMA} user-pair, the duration of the second phase is $t_{k,i}$. Consequently, the achievable offloading data rate at the \ac{BS} for the $i$-th user and $j$-th user can be, respectively, expressed as 
\begin{equation} \label{eq: rate at BS for user-i}
  R_{k,\textnormal{BS}}^{i}=W \log_2 (1 + \frac{\beta_{k,i}p_{k,i} |\mathbf{h}_{k,i}^{H} \mathbf{w}_{k} |^2}{\beta_{k,j}p_{k,i} |\mathbf{h}_{k,i}^{H} \mathbf{w}_{k}|^2 + \sigma^2}),
\end{equation}
\begin{equation} \label{eq: rate at BS for user-j}
   R_{k,\textnormal{BS}}^{j}=W \log_2 (1 + \frac{\beta_{k,j}p_{k,i} |\mathbf{h}_{k,i}^{H} \mathbf{w}_{k}|^2}{\sigma^2}).
\end{equation}
\indent Recall that the duration of the
second phase is $t_{k,i}$. Hence, the energy consumption for data offloading of the $i$-th user is $E_{k,i}^{\textnormal{off}}=t_{k,i} p_{k,i}$. According to~\cite{8488502}, by combining the two phases, the offloaded data that are processed by the \ac{MEC} server pool at the \ac{BS} for $i$-th user and $j$-th user can be, respectively, written as 
\begin{equation}\label{eq: offloading data of user i}
  L_{k,i}^{\textnormal{off}}= t_{k,i} R_{k,\textnormal{BS}}^{i},
\end{equation}
\begin{equation}\label{eq: offloading data of user j}
 L_{k,j}^{\textnormal{off}}= \textnormal{min} \{t_{k,j} R_{k,i}^j, t_{k,i} R_{k,\textnormal{BS}}^{j}\}.
\end{equation}
\section{Total Users' Energy Consumption Minimization}
\label{Sec: The Proposed NOMA-assisted Cooperative THz-SIMO MEC Scheme}
In this section, a joint user pairing, beamforming, time, and power allocation optimization problem that aims to minimize the total energy consumption of all the users in the proposed scheme under the full-offloading requirement of each user is formulated. Somehow similar problems have been investigated in~\cite{8951269} and~\cite{9348649}, but (i) for two users only, (ii) with cooperative \ac{SISO} models, (iii) with a direct link between the cell-edge user and \ac{BS}, and (iv) for the partial-offloading model. These four aforementioned differences between this paper and the related literature make the formulation and solution of the optimization problem in our model unique. 

In the $k$-th \ac{NOMA} user-pair, one can define the energy consumption of the $j$-th cell-edge user and the $i$-th cooperating cell-center user as $E_{k,j}^{\textnormal{off}}+E_{k,i}^{\textnormal{off}}$. As discussed before in Section~\ref{Section: System Model and Problem Formulation}, a common latency constraint is adopted in the proposed system to ensure that all users can receive the results of their offloading data within a given time (i.e., $\sum_{k=1}^{K} (t_{k,j}+t_{k,i}) \leq T$). Consequently, for the $k$-th \ac{NOMA} user-pair, the latency constraint is $t_{k,j}+t_{k,i} \leq T/K$. For the data offloading, the energy consumption depends on both the users' transmit power and time duration for offloading data~\cite{9348649}. Mathematically, the energy consumption minimization problem of the $k$-th \ac{NOMA} user-pair can be defined as 
\begin{alignat}{3}
\textnormal{(P0):} &\min_{\substack{\{\mathbf{\Gamma}, \mathbf{w}_{k}, \\ t_{k,j},t_{k,i}, \\ p_{k,j},p_{k,i}\}}} &\ &E_{k,j}^{\textnormal{off}}+E_{k,i}^{\textnormal{off}}, \ k \in [1, ..., K] \label{eq:objective0}\\
&\quad \ \text{s.t.} &  &  L^{\textnormal{off}}_{k,j} \geq L_{k,j}, \ k \in [1, ..., K] \label{eq:constraint-b0}\\
&  &  & L^{\textnormal{off}}_{k,i} \geq L_{k,i}, \ k \in [1, ..., K] \label{eq:constraint-c0} \\
&  &  & t_{k,j}+t_{k,i} \leq T/K, \ k \in [1, ..., K] \label{eq:constraint-d0} \\
&  &  & L_{k,j}^{\textnormal{off}} >0, L_{k,i}^{\textnormal{off}} > 0, L_{k,j}^{\textnormal{loc}}= 0, \nonumber \\ 
&  &  & L_{k,i}^{\textnormal{loc}} = 0, \ k \in [1, ..., K] \label{eq:constraint-e0}
\end{alignat}
\noindent where $\mathbf{\Gamma}$ denotes a vector that contains all \ac{NOMA} user-pair combinations. By~\eqref{eq:constraint-b0} and~\eqref{eq:constraint-c0} we ensure that all the users' tasks are accomplished remotely in the \ac{MEC} server pool. \eqref{eq:constraint-d0} represents the latency requirement for the $k$-th \ac{NOMA} user-pair. \eqref{eq:constraint-e0} is set as the full-offloading model is adopted. It is worth noting that the above optimization problem is for one \ac{NOMA} user-pair, the optimization of the other \ac{NOMA} user pairs can be performed similarly, and then the total energy consumption is calculated by accumulating the energy consumption of each \ac{NOMA} user-pair.

Clearly, solving (P0) directly using existing optimization tools is infeasible as the problem is non-convex. The non-convexity in (P0) lies in the following, (i) there is a coupling between the time allocation optimization variables ($t_{k,j},t_{k,i}$) and the transmission power optimization variables ($p_{k,j},p_{k,i}$) in the objective function~\eqref{eq:objective0} as well as in the constraints~\eqref{eq:constraint-b0} and~\eqref{eq:constraint-c0}, (ii) there is a coupling, in general, between the beamforming vector ($\mathbf{w}_{k}$) and the transmission power optimization variables ($p_{k,j},p_{k,i}$) in the constraints~\eqref{eq:constraint-b0} and~\eqref{eq:constraint-c0}. Finding the optimal solution for such joint optimization problem using direct search is computationally prohibitive and is not suitable for our delay-sensitive application. Hence, to solve (P0), we propose a suboptimal solution that decomposes (P0) into three sub-problems: (i) a user pairing sub-problem, (ii) a beamforming sub-problem, and (iii) a joint time and power allocation sub-problem, then, we solve them sequentially. This approach is usually adopted in the literature to deal with similar optimization problems as in~\cite{9348649,9417585,9417425}. Algorithm~\ref{Tab: A summary of the solution methodology for the first optimization problem} provides a summary of the proposed solution for (P0), which is discussed next.
\begin{algorithm*}[!t]
\caption{A summary of the proposed solution for (P0)} \label{Tab: A summary of the solution methodology for the first optimization problem}
\justify
(-)~\textbf{Initialization:} (i) distribute the cell-edge users, $K_\textnormal{Edge}$, and cooperating cell-center users, $K_\textnormal{Coop}$, based on the uniform random distribution, (ii) calculate the channel gain of each cooperating cell-center user, $\mathbf{h}_{i} \ \forall i \in [1, ..., K_\textnormal{Coop}]$, based on (1), and (iii) build the BS beamforming vectors set, $\mathbf{w}_b \ \forall b \in [0,B]$, based on (16)\;
\hspace{-1.5em}~(A)~\textbf{Cell-edge and cooperating cell-center user pairing sub-problem:} (i) pair each cell-edge user with each cooperating cell-center user based on the Hungarian algorithm and store all NOMA user-pair combinations in $\mathbf{\Gamma}$, then (ii) in each formed pair, calculate the channel gain between the paired users, $h_{k,ji} \ \forall  k \in [1, ..., K], \ \forall  ji \in \mathbf{\Gamma} $, based on (3)\;
\hspace{-1.5em}~(B)~\textbf{Analog beamforming sub-problem} (i.e., Cooperating cell-center user scheduling)\textbf{:} find the best beamforming vector for each cooperating cell-center user according to its location based on receive analog beamforming with the aid of the cosine similarity metric mentioned in (17)\;
\hspace{-1.5em}~(C)~\textbf{Joint time allocation and power allocation sub-problem:} With the provided user pairing and beamforming from the previous steps, determine the optimal ``transmission time and transmission power allocation'' that minimize the total energy consumption of all the users in the system through the Lagrange duality method provided in Proposition~\ref{proposition-1}\;
\end{algorithm*}

For the user pairing sub-problem, in the beginning, the users in the system are distributed according to the uniform random distribution~\cite{9264161} and then each cell-edge user is paired with a cooperating cell-center user based on the shortest Euclidean distance. To achieve this, the Hungarian algorithm is adopted, which is a well-known one-to-one assignment matching algorithm~\cite{kuhn1955hungarian}. Hungarian algorithm is chosen for two reasons, (i) when the problem is a one-to-one assignment problem (same as our optimization problem), the Hungarian algorithm provides the optimal assignment for the problem~\cite{mills2007dynamic}, and (ii) the computational complexity of the exhaustive-search algorithm is $O(K_\textnormal{Coop}\,!)$, which is much more involved than the adopted Hungarian algorithm that has a computational complexity of $O(K_\textnormal{Coop}^3)$. The detailed algorithm~\cite{Hungarian_algorithm} is omitted for brevity. After conducting the \ac{NOMA} user pairing, the channel gain between each paired cell-edge user and its cooperating cell-center user is obtained through~\eqref{eq: channel gain between user-i and user-j}.

For the beamforming sub-problem, as low hardware cost and power consumption are essential for energy-efficient future wireless networks, the hybrid and analog beamforming schemes, where the number of the \ac{RF} chains at the \ac{BS} is reduced, have gained interest over the last decade~\cite{9264161,maraqa2021energy,9115278,8733134,8706964}. Hence, we adopt here an analog beamforming scheme that has fixed predefined beams for the \ac{BS} to choose from. These beams are evenly distributed over each sector coverage area. So, the cell is assumed to have three $120^{\circ}$ sectors as shown in Fig.~\ref{fig: THz NOMA system model}. With one \ac{RF} chain available at the \ac{BS}, only one beam can be selected at a time, and hence one beam to serve one \ac{NOMA} user-pair per channel use. A time-division strategy is used to alternate among the different \ac{NOMA} user pairs so that \ac{ABF} can only generate one beam at a time. So, considering a specific sector, $\bar\theta$, from $-\pi/6$ to $\pi/2$ or $330^{\circ}$ to $90^{\circ}$, the $120^{\circ}$ area is covered by a set of $B+1$ beams and similarly for the other sectors. Each beam-$b$ of the available beams has the following steering vector~\cite{9264161}:
\begin{equation}\label{eq:Beams}
\mathbf{w}_b = \boldsymbol{a}(\bar\theta_b), \forall b \in [0,B],
\end{equation}
\noindent where the parameter $\bar\theta_b = -\pi/6 + (b \times \frac{2 \pi}{3 B})$. In this way, this entire sector region is divided into $B$ equal-spaced angles, effectively forming a set of $B+1$ beams. The $B+1$ beams can be thought of as a choice of $B+1$ different steering vectors based on~\eqref{eq:Beams}, such that collectively, the steering vectors of the $B+1$ candidate precoding vectors uniformly cover the entire sector region of $\bar\theta = -\pi/6$ to $\pi/2$. 

The cosine similarity metric is utilized to determine the level of correlation among the \ac{THz} cooperating cell-center users and the \ac{BS} available beams that are formed through analog beamforming. Several works in \ac{mmWave}-\ac{NOMA} systems have used the cosine similarity metric to determine the correlations between the users' channels and the fixed beams~\cite{9264161} or among the users' channels~\cite{8454272}. Similarly, this concept can be used herein in the context of \ac{THz}-\ac{NOMA} systems. In particular, we use the result from~\cite{9264161}, where it is shown that the cosine similarity metric between the channel $\mathbf{h}_{k,i}$ of the $i$-th user, in the $k$-th \ac{NOMA} user-pair, and a beam-$b$ with beamforming vector $\mathbf{w}_{b}$ can be expressed as follows:
\begin{equation} \label{eq:cosS}
\begin{split}
\cos(\mathbf{h}_{k,i},\mathbf{w}_{b}) & = \frac{\mid \mathbf{h}_{k,i}^{H} \mathbf{w}_{b} \mid}{||\mathbf{h}_{k,i}||~||\mathbf{w}_{b}||} \\ &=
      \frac{\mid \boldsymbol{a}(\phi_{k,i})^{H} \boldsymbol{a}(\phi_{b}) \mid}{N} \\ &= F_N \big(\pi [\phi_{k,i} - \phi_b]\big),
\end{split}
\end{equation}
\noindent where $\phi_b$~=~$\sin(\bar\theta_b)$ and $\phi_{k,i}$~=~$\sin(\theta_{k,i})$ are the normalized directions of the candidate beam and the user channel, respectively, and $F_N$ represents the Fejer Kernel. The properties of Fejer Kernel dictate that as $|\phi_{k,i} - \phi_b|$ increases, $\cos(\mathbf{h}_{{k,i}},\mathbf{w}_{b})\rightarrow 0$. So, each cooperating cell-center user in a \ac{NOMA} user-pair is scheduled with its best \ac{THz} beam by guaranteeing that the selected beam and the cooperating cell-center user direction are well aligned, as reflected through the large value of the cosine similarity metric. 

Next, the time allocation and power allocation are jointly optimized and the solution is obtained through the Lagrange duality method~\cite{9348649}. With the given user pairing and beamforming from the previous problems, (P0) can be reformulated as
\begin{alignat}{3}
\textnormal{(P1):} &\min_{\substack{\{t_{k,j},t_{k,i}, \\p_{k,j},p_{k,i}\}}} &\ &E_{k,j}^{\textnormal{off}}+E_{k,i}^{\textnormal{off}}, \ k \in [1, ..., K] \label{eq:objective1}\\
&\quad \ \text{s.t.} &  &  L^{\textnormal{off}}_{k,j} \geq L_{k,j}, \ k \in [1, ..., K] \label{eq:constraint-b1}\\
&  &  & L^{\textnormal{off}}_{k,i} \geq L_{k,i}, \ k \in [1, ..., K] \label{eq:constraint-c1} \\
&  &  & t_{k,j}+t_{k,i} \leq T/K, \ k \in [1, ..., K] \label{eq:constraint-d1} \\
&  &  & L_{k,j}^{\textnormal{off}} >0, L_{k,i}^{\textnormal{off}} > 0, L_{k,j}^{\textnormal{loc}}= 0, \nonumber \\
&  &  & L_{k,i}^{\textnormal{loc}} = 0, \ k \in [1, ..., K] \label{eq:constraint-e1}
\end{alignat}
\noindent where the optimal solution of (P1) is provided in Proposition~\ref{proposition-1}.  

\begin{proposition}\label{proposition-1}
The optimal time and power allocation that minimize the energy consumption of one \ac{NOMA} user-pair are given by
\begin{equation} \label{eq: optimal values of t1 t2 p1 p2}
\begin{gathered}
  t_{k,j}^*=\frac{T}{K}-\frac{L_{k,i}}{W \log_2 (1 +\frac{\beta_{k,i}}{\beta_{k,j}})} , \ t_{k,i}^*=\frac{L_{k,i}}{W \log_2 (1 +\frac{\beta_{k,i}}{\beta_{k,j}})}, \\
p_{k,j}^*=\frac{\sigma^2(2^{L_{k,j}/W t_{k,j}^*}-1)}{|h_{k,ji}|^2} , \ p_{k,i}^*=\frac{\sigma^2(2^{L_{k,j}/W t_{k,i}^*}-1)}{\beta_{k,j} |\mathbf{h}_{k,i}^{H} \mathbf{w}_{k}|^2}.
\end{gathered}
\end{equation}
\end{proposition}
\begin{proof}
See the Appendix.  
\end{proof}

\section{Total Users' Computation Energy Efficiency (CEE) Maximization}
\label{Section: Total Users' Computation Energy Efficiency Maximization}

\begin{algorithm*}[!t] 
\caption{A summary of the proposed solution for (P4)} \label{Tab: A summary of the solution methodology for the second optimization problem}
\justify
(-)~\textbf{Initialization:} (i) distribute the cell-edge users, $K_\textnormal{Edge}$, and cooperating cell-center users, $K_\textnormal{Coop}$, based on the uniform random distribution, and (ii) calculate the channel gain of each cooperating cell-center user, $\mathbf{h}_{i} \ \forall i \in [1, ..., K_\textnormal{Coop}]$, based on (1)\;
\hspace{-1.5em}~(A)~\textbf{Cell-edge and cooperating cell-center user pairing sub-problem:} (i) pair each cell-edge user with each cooperating cell-center user based on the Hungarian algorithm and store all NOMA user-pair combinations in $\mathbf{\Gamma}$, then (ii) in each formed pair, calculate the channel gain between the paired users, $h_{k,ji} \ \forall  k \in [1, ..., K], \ \forall  ji \in \mathbf{\Gamma} $, based on (3)\;
\hspace{-1.5em}~(B)~\textbf{Time allocation sub-problem:} divide the whole transmission time block $T$ (i.e., the latency requirement of the network) equally between all the formulated NOMA user pairs (i.e., $ t_{k,j}+t_{k,i} = T/K, \ \forall k \in [1, ..., K], \ \forall  ji \in \mathbf{\Gamma}$)\;
\hspace{-1.5em}~(C)~\textbf{Power and beam-gain allocation sub-problem:} using an alternating approach, with the aid of the Dinkelbach algorithm, find an upper bound solution for both the power allocation variables (i.e., $\{\overline{p}_{k,j}, \overline{p}_{k,i}\}, \ \forall k \in [1, ..., K], \ \forall  ji \in \mathbf{\Gamma}$) and the beam gain variable (i.e., $\overline{c}_{k,i}, \ \forall k \in [1, ..., K], \ \forall  ji \in \mathbf{\Gamma}$) without considering the CM constraint\;
\hspace{-1.5em}~(D)~\textbf{Beamforming sub-problem:} while considering the CM constraint and the upper bound solution $\{\overline{p}_{k,j},\overline{p}_{k,i},$ $\overline{c}_{k,i}\}, \ \forall k \in [1, ..., K], \ \forall  ji \in \mathbf{\Gamma}$ from (C), and with the aid of the Dinkelbach algorithm, find the actual beamforming vectors $\mathbf{w}^*_{k}, \ \forall k \in [1, ..., K]$ and the power allocation variables $\{p^*_{k,j},p^*_{k,i}\}, \ \forall k \in [1, ..., K]$ that maximize (P4)\;
\end{algorithm*}

In this section, a joint user pairing, time, power, and beamforming allocation optimization problem that aims to maximize the total users' \ac{CEE} in the proposed scheme under the full-offloading requirement of each user is formulated. In the following, we provide the \ac{CEE} formula of the $k$-th \ac{NOMA} user-pair that can be defined as the ``ratio between the total offloaded data bits and the users' energy consumption (Bits/Joule/Hz)" as~\cite{9345931}
\begin{equation} \label{eq: CEE formula}
\begin{split}
  \eta_{k,\textnormal{CEE}} &= \frac{ L_{k,i}^{\textnormal{off}}  + L_{k,j}^{\textnormal{off}}}{E_{k,j}^{\textnormal{off}} + E_{k,i}^{\textnormal{off}}} = \frac{ t_{k,i} R_{k,\textnormal{BS}}^i  + \textnormal{min}\{t_{k,j} R_{k,i}^{j}, t_{k,i} R_{k,\textnormal{BS}}^j\}}{t_{k,j} p_{k,j} + t_{k,i} p_{k,i}}\\
\end{split}
\end{equation}

Similar to the users' energy consumption minimization problem discussed in Section~\ref{Sec: The Proposed NOMA-assisted Cooperative THz-SIMO MEC Scheme}, we adopt the \ac{ABF} structure with a single \ac{RF} chain. Differently, we conduct a formal beamforming optimization here, to do that we define the beamforming vector for the $k$-th \ac{NOMA} user-pair as $\mathbf{w}_{k}$, in which $|[\mathbf{w}_{k}]_n|=\frac{1}{\sqrt{N}}$, $ n \in [1, ..., N]$, where $[\mathbf{w}_{k}]_n$ is the $n$-th element of $\mathbf{w}_{k}$. This is called the \ac{CM} constraint~\cite{8706964}. Mathematically, the \ac{CEE} optimization problem of the $k$-th \ac{NOMA} user-pair can be defined as follows:
\begin{alignat}{3}
\textnormal{(P4):} &\max_{\substack{\{\mathbf{\Gamma}, t_{k,j}, \\ t_{k,i},p_{k,j}, \\ p_{k,i},\mathbf{w}_{k}\}}} &\ &\eta_{k,\textnormal{CEE}}, \ k \in [1, ..., K] \label{eq:objectiveEE}\\
&\quad \ \text{s.t.} &  &  L^{\textnormal{off}}_{k,j} \geq \frac{L_{k,j}}{W}, \ k \in [1, ..., K] \label{eq:constraint-EE1}\\
&  &  & L^{\textnormal{off}}_{k,i} \geq \frac{L_{k,i}}{W}, \ k \in [1, ..., K] \label{eq:constraint-EE2} \\
&  &  & L_{k,j}^{\textnormal{off}} >0, L_{k,i}^{\textnormal{off}} > 0, L_{k,j}^{\textnormal{loc}}= 0, \nonumber \\ 
&  &  & L_{k,i}^{\textnormal{loc}} = 0, \ k \in [1, ..., K] \label{eq:constraint-EE3} \\
&  &  & t_{k,j}+t_{k,i} \leq T/K, \ k \in [1, ..., K] \label{eq:constraint-EE4} \\
&  &  & 0 < p_{k,i} \leq p_{\textnormal{max}}, 0 < p_{k,j} \leq p_{\textnormal{max}}, \ k \in [1, ..., K] \label{eq:constraint-EE5} \\
&  &  & |[\mathbf{w}_{k}]_n|=\frac{1}{\sqrt{N}}, \ n \in [1, ..., N], \ k \in [1, ..., K] \label{eq:constraint-EE6}
\end{alignat}
\noindent where $\mathbf{\Gamma}$ denotes again a vector that contains all \ac{NOMA} user-pair combinations. By~\eqref{eq:constraint-EE1},~\eqref{eq:constraint-EE2},~\eqref{eq:constraint-EE3} we ensure the full-offloading requirement of each user. \eqref{eq:constraint-EE4} represents the latency requirement for the $k$-th \ac{NOMA} user-pair. \eqref{eq:constraint-EE5} is set to ensure that the users' transmit power do not exceed the maximum allowed transmit power $p_{\textnormal{max}}$. \eqref{eq:constraint-EE6} denotes the \ac{CM} constraint. It is worth noting that the above optimization problem is for one \ac{NOMA} user-pair, the optimization of the other \ac{NOMA} user pairs can be performed similarly, and then the total \ac{CEE} of all \ac{NOMA} user pairs is calculated by accumulating the \ac{CEE} of each \ac{NOMA} user-pair. 

Clearly, solving (P4) directly using existing optimization tools is infeasible as the problem is non-convex (i.e., \eqref{eq:objectiveEE} and \eqref{eq:constraint-EE6} are non-convex). The challenges here lie in the followings, (i) the objective function is fractional and includes the beamforming vector, (ii) there is a coupling between the beamforming vector ($\mathbf{w}_{k}$) and the transmission power optimization variables ($p_{k,j},p_{k,i}$) in~\eqref{eq:objectiveEE},~\eqref{eq:constraint-EE1}, and~\eqref{eq:constraint-EE2}, (iii) there is also a coupling between the time allocation optimization variables ($t_{k,j},t_{k,i}$) and the transmission power optimization variables ($p_{k,j},p_{k,i}$) in~\eqref{eq:objectiveEE},~\eqref{eq:constraint-EE1}, and~\eqref{eq:constraint-EE2}, and (iv) finding the optimal solution for such joint optimization problem using direct search is computationally prohibitive and is not suitable for our delay-sensitive application. Hence, to solve (P4), we propose a suboptimal solution that decomposes (P4) into four problems: (i) user pairing problem, (ii) time allocation problem, (iii) beam gain allocation and power control problem, and (iv) beamforming problem, then, we solve them sequentially. This approach is usually adopted in the literature to deal with similar optimization problems as in~\cite{8294044,9145398,8415781,9345931}. Algorithm~\ref{Tab: A summary of the solution methodology for the second optimization problem} provides a summary of the proposed solution for (P4), which is discussed next.

Firstly, similar to the energy consumption minimization problem discussed in Section~\ref{Sec: The Proposed NOMA-assisted Cooperative THz-SIMO MEC Scheme}, Hungarian Algorithm is utilized here to pair each cell-edge user with a cooperating cell-center user. Secondly, the time allocation for (P4) can be obtained when $ t_{k,j}+t_{k,i} = T/K, \ k \in [1, ..., K]$. This can be proved by a contradiction approach similar to the proof provided in Appendix A of~\cite{9345931}. As the half-duplex cooperation is adopted, in each \ac{NOMA} pair, the transmission time for the data offloading in the links between the cell-edge user and the cell-center user and the link between the cell-center user and the \ac{BS} is divided equally. 

Thirdly, using an alternating approach we can solve the beam gain allocation and power control sub-problem as follows. Let us denote that beam gain of the $i$-th user as $c_{k,i} = |\mathbf{h}_{k,i}^{H} \mathbf{w}_{k}|^2$. Without considering the \ac{CM} constraint (i.e., with the ideal beamforming), the beam gain of the $i$-th user satisfies~[30, Lemma 1] 
\begin{equation} \label{eq: beam gain}
 \frac{c_{k,i}}{|\lambda_{k,i}|^2}=N, \textnormal{where} \ \lambda_{k,i}=\sqrt{\frac{G_{t,i} G_{r,\textnormal{BS}}}{\mathcal{PL}(f,d_i)}}
\end{equation}
Now, by substituting~\eqref{eq: beam gain} in (P4) and by considering the ideal beamforming, the beam gain allocation and power control problem can be expressed as in (P5).

\indent In (P5), there is still coupling between the power allocation variables (i.e., $\{p_{k,j},p_{k,i}\}$) and the beam gain variable (i.e., $c_{k,i}$). To tackle this, similar to~\cite{8415781,9145398}, we resort to an alternating approach that can iteratively optimize these variables, as follows:

\textit{1) Updating $\{p_{k,j},p_{k,i}\}$ under a given $c_{k,i}$}

To get the initial value for $c_{k,i}$, similar to~\cite{8415781,9145398}, instead of solving the \ac{CEE} maximization problem, we can solve its related sum rate maximization problem. The initial value for $c_{k,i}$ can be obtained by following the same steps in Section III.B in~\cite{8415781} and can be expressed as 

\begin{strip}
\begin{bottomborder} \end{bottomborder}

\begin{alignat}{3}
\textnormal{(P5):} &\max_{\{p_{k,j}, p_{k,i}, c_{k,i}\}} & \ & \frac{ t_{k,i} \log_2 (1 + \frac{\beta_{k,i}p_{k,i} c_{k,i} }{\beta_{k,j}p_{k,i} c_{k,i} + \sigma^2})  + \textnormal{min}\{t_{k,j} \log_2 (1 + \frac{p_{k,j} |h_{k,ji}|^2}{\sigma^2}), t_{k,i} \log_2 (1 + \frac{\beta_{k,j}p_{k,i} c_{k,i}}{\sigma^2}) \}}{t_{k,j} p_{k,j} + t_{k,i} p_{k,i}}, \label{eq:objective_power_subproblem}\\
&\quad \ \text{s.t.} & & \eqref{eq:constraint-EE1},~\eqref{eq:constraint-EE2},~\eqref{eq:constraint-EE3},~\eqref{eq:constraint-EE4},~\eqref{eq:constraint-EE5},~\eqref{eq: beam gain},     \label{eq:constraint-EE1simplified}
\end{alignat}

\begin{alignat}{3}
\textnormal{(P6):} &\underset{\{p_{k,j},p_{k,i}\}}{\text{max}} &\ &\frac{ t_{k,i} \log_2 (1 + \frac{\beta_{k,i}p_{k,i} c_{k,i} }{\beta_{k,j}p_{k,i} c_{k,i} + \sigma^2})  + \textnormal{min}\{t_{k,j} \log_2 (1 + \frac{p_{k,j} |h_{k,ji}|^2}{\sigma^2}), t_{k,i} \log_2 (1 + \frac{\beta_{k,j}p_{k,i} c_{k,i}}{\sigma^2}) \}}{t_{k,j} p_{k,j} + t_{k,i} p_{k,i}}, \label{eq:objective_power}\\
&\quad \ \text{s.t.} &  &  t_{k,j} \log_2 (1 + \frac{p_{k,j} |h_{k,ji}|^2}{\sigma^2}) \geq \frac{L_{k,j}}{W},  \ k \in [1, ..., K] \label{eq:constraint-power1}\\
&  &  & t_{k,i} \log_2 (1 + \frac{\beta_{k,j}p_{k,i} c_{k,i}}{\sigma^2}) \geq \frac{L_{k,j}}{W},  \ k \in [1, ..., K] \label{eq:constraint-power2} \\
&  &  & t_{k,i} \log_2 (1 + \frac{\beta_{k,i}p_{k,i} c_{k,i}}{\beta_{k,j}p_{k,i} c_{k,i} + \sigma^2}) \geq \frac{L_{k,i}}{W},  \ k \in [1, ..., K] \label{eq:constraint-power3} \\
&  &  & \eqref{eq:constraint-EE3},~\eqref{eq:constraint-EE5}, \label{eq:constraint-power4} 
\end{alignat}

\begin{alignat}{3}
\textnormal{(P7):} &\underset{\{p_{k,j},p_{k,i}\}}{\text{max}} &\ &\frac{ t_{k,i} \log_2 (1 + \frac{\beta_{k,i}p_{k,i} c_{k,i} }{\beta_{k,j}p_{k,i} c_{k,i} + \sigma^2})  + \textnormal{min}\{t_{k,j} \log_2 (1 + \frac{p_{k,j} |h_{k,ji}|^2}{\sigma^2}), t_{k,i} \log_2 (1 + \frac{\beta_{k,j}p_{k,i} c_{k,i}}{\sigma^2}) \}}{t_{k,j} p_{k,j} + t_{k,i} p_{k,i}},  \label{eq:objective_power_relaxed}\\
&\quad \ \text{s.t.} &  & \eqref{eq:constraint-EE3},~\eqref{eq:constraint-EE5},~\eqref{eq:affine1},~\eqref{eq:affine2},~\eqref{eq:affine3},
\end{alignat} 

\begin{alignat}{3}
\textnormal{(P8):} & \quad \underset{c_{k,i}}{\text{max}} &\ &\frac{ t_{k,i} \log_2 (1 + \frac{\beta_{k,i}p_{k,i} c_{k,i} }{\beta_{k,j}p_{k,i} c_{k,i} + \sigma^2})  + \textnormal{min}\{t_{k,j} \log_2 (1 + \frac{p_{k,j} |h_{k,ji}|^2}{\sigma^2}), t_{k,i} \log_2 (1 + \frac{\beta_{k,j}p_{k,i} c_{k,i}}{\sigma^2}) \}}{t_{k,j} p_{k,j} + t_{k,i} p_{k,i}}, \label{eq:objective_beam_relaxed}\\
&\quad \ \text{s.t.} &  & \eqref{eq:constraint-EE3},~\eqref{eq: beam gain},~\eqref{eq:affine1},~\eqref{eq:affine2},~\eqref{eq:affine3},
\end{alignat} 

\begin{bottomborder} \end{bottomborder}
\end{strip}

\begin{equation}
    c_{k,i} = \frac{(2^{\frac{L_{k,j}}{t_{k,i} W}}-1)(\sigma^2)}{p_\textnormal{max}}. \label{eq:initial_gain} 
\end{equation}
Now, (P5) can be simplified as in (P6),
\noindent where through simple manipulations~\eqref{eq:constraint-power1},~\eqref{eq:constraint-power2}, and~\eqref{eq:constraint-power3} can be transformed into the following affine constraints
\begin{equation}
    p_{k,j} |h_{k,ji}|^2 \geq (2^{\frac{L_{k,j}}{t_{k,j} W}}-1)(\sigma^2), \label{eq:affine1} 
\end{equation}
\begin{equation}
    \beta_{k,j}p_{k,i} c_{k,i} \geq (2^{\frac{L_{k,j}}{t_{k,i} W}}-1)(\sigma^2), \label{eq:affine2}
\end{equation}
\begin{equation}
    \beta_{k,i}p_{k,i} c_{k,i} \geq (2^{\frac{L_{k,i}}{t_{k,i} W}}-1)(\beta_{k,j}p_{k,i} c_{k,i} + \sigma^2), \label{eq:affine3}
\end{equation}
\noindent Next, we can re-write (P6) by replacing~\eqref{eq:constraint-power1},~\eqref{eq:constraint-power2}, and~\eqref{eq:constraint-power3} with~\eqref{eq:affine1},~\eqref{eq:affine2}, and~\eqref{eq:affine3}, respectively, as in (P7).
The resultant optimization problem (P7) is quasi-concave, as (i) its objective function is fractional with a concave numerator and an affine denominator, and (ii) its constraints are all affine. Hence, (P7) can be solved by Dinkelbach algorithm~\cite{dinkelbach1967nonlinear} with the aid of the CVX in Matlab. At this stage, $\{p_{k,j},p_{k,i}\}$ are calculated and as discussed next these values are used to update $c_{k,i}$.

\textit{2) Updating $c_{k,i}$ under given $\{p_{k,j},p_{k,i}\}$}\\
In a similar manner to (P7), (P5) can be simplified as in (P8).
The resultant optimization problem (P8) is a convex problem, as (i) its objective function is concave, and (ii) its constraints are all affine. Hence, (P8) can be solved using standard convex optimization tools~\cite{boyd2004convex}.
Finally, the solution of $\{p_{k,j},p_{k,i},c_{k,i}\}$ can be found by repeating step 1) and step 2) until convergence. The convergence of this approach can be proved in a similar way to the proof of Theorem 2 in~\cite{9145398}. Let us denotes the obtained solution of (P5) as $\{\overline{p}_{k,j}, \overline{p}_{k,i}, \overline{c}_{k,i}\}$. It is important to note that this solution can be considered as an upper-bound solution for the original problem (P4), since this solution does not consider the \ac{CM} constraint~\cite{9145398}.

Fourthly, here, we consider the beamforming sub-problem that takes into account the \ac{CM} constraint and aims to design suitable $\mathbf{w}_{k} \in \mathbb{C}^N, \ k \in [1, ..., K]$ while utilizing the obtained $\overline{c}_{k,i}$ from the third problem. This beamforming problem can be expressed as 
\begin{alignat}{3}
 & \textnormal{(P9):} &\quad \ & \mathbf{w}_{k} \in \mathbb{C}^N, \ k \in [1, ..., K] \label{eq:objective_BF_subproblem}\\
& \ \text{s.t.} &  & c_{k,i} = |\mathbf{h}_{k,i}^{H} \mathbf{w}_{k}|^2,     \label{eq:constraint_BF_subproblem_1} \\
&  &  & |[\mathbf{w}_{k}]_n|=\frac{1}{\sqrt{N}}. \ n \in [1, ..., N], \ k \in [1, ..., K] \label{eq:constraint_BF_subproblem_2}
\end{alignat}
\indent It is worth noting that with known $\{p_{k,j},p_{k,i}\}$, instead of solving the \ac{CEE} problem, its related sum rate maximization problem can be solved in a similar way to~\cite{9145398}. While the value of $c_{k,i}$ maximizes the sum rate, it also affects the full-offloading constraints (i.e.,~\eqref{eq:affine2} and~\eqref{eq:affine3}) and that these two constraints should not be violated. Specifically, $c_{k,i} \geq \frac{(2^{\frac{L_{k,j}}{t_{k,i} W}}-1)(\sigma^2)}{\beta_{k,j}p_{k,i}}$ and $c_{k,i} \geq \frac{ (2^{\frac{L_{k,i}}{t_{k,i} W}}-1)(\sigma^2)}{p_{k,i} (1-\beta_{k,j}2^{\frac{L_{k,i}}{t_{k,i} W}})}$. Consequently, (P9) can be transformed into 
\begin{alignat}{3}
\textnormal{(P10):} & \quad \underset{\mathbf{w}_{k}}{\text{max}} &\ &|\mathbf{h}_{k,i}^{H} \mathbf{w}_{k}|^2, \ k \in [1, ..., K] \label{eq:objective_BF_relaxed}\\
& \quad \ \text{s.t.} &  &  |\mathbf{h}_{k,i}^{H} \mathbf{w}_{k}|^2 \geq \frac{(2^{\frac{L_{k,j}}{t_{k,i} W}}-1)(\sigma^2)}{\beta_{k,j}p_{k,i}}, \ k \in [1, ..., K]     \label{eq:constraint_BF_relaxed_1} \\
&  &  & |\mathbf{h}_{k,i}^{H} \mathbf{w}_{k}|^2 \geq \frac{ (2^{\frac{L_{k,i}}{t_{k,i} W}}-1)(\sigma^2)}{p_{k,i} (1-\beta_{k,j}2^{\frac{L_{k,i}}{t_{k,i} W}})}, \ k \in [1, ..., K] \label{eq:constraint_BF_relaxed_2} \\
&  &  & |[\mathbf{w}_{k}]_n|=\frac{1}{\sqrt{N}}. \ n \in [1, ..., N], \ k \in [1, ..., K] \label{eq:constraint_BF_relaxed_3}
\end{alignat} 

\begin{strip}
\begin{bottomborder} \end{bottomborder}
\begin{alignat}{3}
\textnormal{(P12):} &\underset{\{p_{k,j},p_{k,i}\}}{\text{max}} &\ &\frac{ t_{k,i} \log_2 (1 + \frac{\beta_{k,i}p_{k,i} c^*_{k,i} }{\beta_{k,j}p_{k,i} c^*_{k,i} + \sigma^2})  + \textnormal{min}\{t_{k,j} \log_2 (1 + \frac{p_{k,j} |h_{k,ji}|^2}{\sigma^2}), t_{k,i} \log_2 (1 + \frac{\beta_{k,j}p_{k,i} c^*_{k,i}}{\sigma^2}) \}}{t_{k,j} p_{k,j} + t_{k,i} p_{k,i}}, \label{eq:objective_Final_power_relaxed}\\
&\quad \ \text{s.t.} &  & \eqref{eq:constraint-EE3},~\eqref{eq:constraint-EE5},~\eqref{eq:affine1},~\eqref{eq:affine2},~\eqref{eq:affine3},
\end{alignat}
\begin{bottomborder} \end{bottomborder}
\end{strip}

\indent Clearly (P10) is non-convex. In order to convexify this problem, (i) the equality constraint in~\eqref{eq:constraint_BF_relaxed_3} can be relaxed to an inequality constraint (which is convex) without affecting the solution as proved in Theorem 2 of~\cite{8415781}, and (ii) the norm operation in~\eqref{eq:objective_BF_relaxed} can be eliminated as introducing an arbitrary rotation to $\mathbf{w}_{k}$ (i.e., if $\mathbf{w}^*_{k}$ is optimal, then $\mathbf{w}^*_{k} e^{j \phi}$ is also optimal, where $\phi$ is an arbitrary phase within $[0,2 \pi)$) does not affect the beam gains~\cite{8415781,8294044}. Specifically, we can choose $\phi$ such that $\mathbf{h}_{k,i}^{H} \mathbf{w}_{k}$ is real-valued. With this, (P10) can be relaxed into the following convex problem 
\begin{alignat}{3}
\textnormal{(P11):} & \quad \underset{\mathbf{w}_{k}}{\text{max}} &\ &\mathbf{h}_{k,i}^{H} \mathbf{w}_{k}, \ k \in [1, ..., K] \label{eq:objective_BF2_relaxed}\\
& \quad \ \text{s.t.} &  &  \textnormal{Im} (\mathbf{h}_{k,i}^{H} \mathbf{w}_{k})=0, \ k \in [1, ..., K]   \label{eq:constraint_BF2_relaxed_1} \\
&  &  & \textnormal{Re}(\mathbf{h}_{k,i}^{H} \mathbf{w}_{k}) \geq \sqrt{\frac{(2^{\frac{L_{k,j}}{t_{k,i} W}}-1)(\sigma^2)}{\beta_{k,j}p_{k,i}}}, \nonumber \\
&  &  & \ k \in [1, ..., K]  \label{eq:constraint_BF2_relaxed_2} \\
&  &  & \textnormal{Re}(\mathbf{h}_{k,i}^{H} \mathbf{w}_{k}) \geq \sqrt{\frac{ (2^{\frac{L_{k,i}}{t_{k,i} W}}-1)(\sigma^2)}{p_{k,i} (1-\beta_{k,j}2^{\frac{L_{k,i}}{t_{k,i} W}})}}, \nonumber \\ 
&  &  & \ k \in [1, ..., K] \label{eq:constraint_BF2_relaxed_3} \\
&  &  & |[\mathbf{w}_{k}]_n| \leq \frac{1}{\sqrt{N}}, \ n \in [1, ..., N], \ k \in [1, ..., K] \label{eq:constraint_BF2_relaxed_4}
\end{alignat} 
\noindent where (P11) can be solved using standard convex optimization tools~\cite{boyd2004convex}. Let us denote the obtained beamforming vector as $\mathbf{w}^*_{k}$. These vectors are used to update the beam gain $c^*_{k,i}=|\mathbf{h}_{k,i}^{H} \mathbf{w}^*_{k}|^2$. Finally, we can use the obtained $c^*_{k,i}$ to get the corresponding power allocation variables $\{p^*_{k,j},p^*_{k,i}\}$ that maximize the original \ac{CEE} problem. To do so, we solve (P12) in the same way we previously solved (P7).

\section{Results and Discussions}
\label{Section: Results and discussions}

In this section, we present detailed numerical results to evaluate the performance of the proposed \ac{NOMA}-assisted cooperative \ac{THz}-\ac{SIMO} \ac{MEC} system for the two energy-related optimizing problems. Specifically, we present the total users' energy consumption minimization-related results followed by total users' \ac{CEE} maximization-related results. In these performance curves, Monte-Carlo simulations that are averaged over $100$ different users' location realizations are used. A list of the default used parameters is provided in Table~\ref{Tab: Simulation Parameters}. Beyond this list, the noise power is $\sigma^2=10\,\textnormal{log}_{10}(W)+~N_f-~174$ dBm, with noise figure $N_f=10$ dB~\cite{9264161}. For comparison purposes, we consider the following baseline schemes:

1) \textbf{THz-NOMA scheme with partial-offloading model and without offloading:} As discussed in the introduction, in the partial-offloading model, the user tasks are partitioned into a local computing part and a remote offloading part. While for the system without offloading the tasks are only locally computed at the user-side. The energy consumption analysis for these baseline schemes are provided next. 
\begin{itemize}[label=$\blacksquare$]
\item \textit{Local computing for the partial-offloading model}
\end{itemize}

\begin{table}[!t]
\centering
\caption{SYSTEM PARAMETERS}
\label{Tab: Simulation Parameters}
\resizebox{0.5\textwidth}{!}{%
\begin{tabular}{|l|l|}
\hline
Parameter name, notation & Value \\ \hline

CPU cycles for each one bit at the users, $\xi_j$, $\xi_i$ & $1$ cycle/bit~\cite{8951269}\\
Effective capacitance coefficient for the users, $\kappa_i$, $\kappa_j$ & $10^{-27}$~\cite{8488502}\\
The length of the time block, $T$ & $0.25$ second~\cite{8488502}\\
Number of input computation bits per user, ($L_{k,j}$, $L_{k,i}$) & ($1,1$) Gbits\\

Number of antennas at the \ac{BS}, $N$ & $4$~\cite{9264161}\\ 

The user's antenna transmission and reception gains, (e.g., $G_{t,i}$ and $G_{r,i}$) & $3$ dBi~\cite{lai2014energy}\\

\ac{BS} antenna reception gain, $G_{r,\textnormal{BS}}$ & $26$ dBi~\cite{8733134}\\

Number of available beams, $B$ & $20$~\cite{maraqa2021energy} \\

Beam-width angle, $\bar\theta_b$ & $6^{\circ}$~\cite{9295330}\\
\ac{BS} sector angular coverage, $\bar\theta$ &$120^{\circ}$~\cite{maraqa2021energy} \\

Power allocation factor per cell-edge user, $\beta_{k,j}$ & $0.3$~\cite{9348649}\\

Power allocation factor per cell-center user, $\beta_{k,i}$ & $0.7$~\cite{9348649}\\

The maximum allowed transmit power, $p_{\textnormal{max}}$ & $9$ dB\\

Water vapor molecules absorption coefficient, $k_{\textnormal{abs}}(f)$ & $0.28 \  \textnormal{m}^{-1}$~\cite{gordon2017hitran2016}\\

Center frequency of the considered \ac{THz} window, & $3.42$ THz~\cite{singh2020analytical} \\
Contiguous bandwidth of the considered \ac{THz} window, $W$ & $137$ GHz~\cite{singh2020analytical}\\

Speed of light, $c$ & $3 \times 10^8$ m/s\\

Number of users, $K_\textnormal{Coop}+K_\textnormal{Edge}$ & $4-20$ users~\cite{9115278}\\
User distribution & Uniform random~\cite{9264161}\\
\ac{BS} coverage region, $d_{\textnormal{Center}}+d_{\textnormal{Edge}}$ & $3+2=5$ meters \\
The initialization parameters of the Dinkelbach algorithm& $\lambda_n=0.01$, $\epsilon=10^{-5}$ \\
\hline
\end{tabular}%
}
\end{table}

Here, we provide the analysis of the energy consumption at the users that arises from the local computation of the tasks that are not offloaded. The local computation at users can be executed during the whole time duration $T$. In the $k$-th \ac{NOMA} user-pair, we can express the local computing energy consumption for the $j$-th user as~\cite{9348649,8488502} 
\begin{equation} \label{eq: local computing of user j}
  E_{k,j}^{\textnormal{loc}}= \sum_{a=1}^{\xi_j L_{k,j}^{\textnormal{loc}}} \kappa_j f_{j,a}= \frac{\kappa_j \xi_j^3 (L_{k,j}^{\textnormal{loc}})^3}{(t_{k,j}+ t_{k,i})^2},
\end{equation}
\begin{equation} \label{eq: CPU frequency}
  f_{j,a}=\frac{\xi_j L_{k,j}^{\textnormal{loc}}}{t_{k,j}+ t_{k,i}}, \forall a \in \{ 1, ..., \xi_j L_{k,j}^{\textnormal{loc}}\},
\end{equation}
\noindent where $\xi_j$ is the number of \ac{CPU} cycles needed to execute one bit in a task at the $j$-th user. $\kappa_j$ is the effective capacitance coefficient at the $j$-th user. $f_{j,a}$ is the \ac{CPU} frequency at the $a$-th cycle, where $a \in \{ 1, ..., \xi_j L_{k,j}^{\textnormal{loc}}\}$. In~\eqref{eq: CPU frequency}, to facilitate the calculations here and similar to~\cite{9348649,8488502}, we assume that the \ac{CPU} frequencies of the different \ac{CPU} cycles are identical. In a similar manner, we can write the local computing energy consumption for the $i$-th user as~\cite{9348649,8488502}  
\begin{equation} \label{eq: local computing of user i}
  E_{k,i}^{\textnormal{loc}}= \frac{\kappa_i \xi_i^3 (L_{k,i}^{\textnormal{loc}})^3}{(t_{k,j}+ t_{k,i})^2}.
\end{equation}
\indent Overall, the energy consumption for the data offloading and the local computing of users data in the $k$-th \ac{NOMA} user-pair is $E_{k,i}^{\textnormal{loc}}+E_{k,i}^{\textnormal{off}}+E_{k,j}^{\textnormal{loc}}+E_{k,j}^{\textnormal{off}}$. This formula is substituted in (P1) as an objective function, then we followed the same steps mentioned in the Appendix to find the optimal time and power allocation variables that minimize the total energy consumption of all the users in the system.
\begin{itemize}[label=$\blacksquare$]
\item \textit{Local computing for the system without offloading}
\end{itemize}

In this scheme, as there is no data offloading $E_{k,j}^{\textnormal{off}}$ and $E_{k,i}^{\textnormal{off}}$ are set to zero. Both $E_{k,j}^{\textnormal{loc}}$ and $E_{k,i}^{\textnormal{loc}}$ are calculated based on~\eqref{eq: local computing of user j} and~\eqref{eq: local computing of user i}, respectively.

2) \textbf{mmWave-NOMA scheme with full-offloading model:} In this baseline scheme, we set the center frequency of the considered \ac{mmWave} channel to $f_\textnormal{mmWave}=28$ GHz with a contiguous bandwidth of $W_\textnormal{mmWave}=2$~GHz~\cite{8454272}. Also, we set the noise power to $\sigma^2_{\textnormal{mmWave}}=- 40$ dBm.

3) \textbf{THz-OMA scheme with full-offloading model:} In this baseline scheme, we assume that the transmission is performed in three phases. Similar to the analysis in Section~\ref{Section: System Model and Problem Formulation}, let us focus on the $k$-th user-pair, and all other user pairs can be analyzed similarly. During the first phase, the $j$-th cell-edge user transmits its offloading data to the $i$-th cell-center user within $t_{k,j}^{\textnormal{OMA}}$. During the second phase, the $i$-th cell-center user transmits its offloading data to the \ac{BS} within $t_{k,i:\textnormal{Phase2}}^{\textnormal{OMA}}$. Finally, during the third phase, the $i$-th cell-center user transmits the received cell-edge's offloading data to the \ac{BS} within $t_{k,i:\textnormal{Phase3}}^{\textnormal{OMA}}$. Subsequently, one can express the achievable offloading data rate of the $j$-th user at the $i$-th user and the achievable offloading data rate at the \ac{BS} for the $i$-th user and $j$-th user, respectively,  as~\cite{tse2005fundamentals}
\begin{equation} \label{eq:  OMA rate at user-i}
  R_{k,i}^{j,\textnormal{OMA}}=W \log_2 (1 + \frac{p_{k,j}^{\textnormal{OMA}} |h_{k,ji}|^2}{\sigma^2}),
\end{equation}
\begin{equation} \label{eq:  OMA rate at BS for user-i}
  R_{k,\textnormal{BS}}^{i,\textnormal{OMA}}= \textnormal{0.5} W \log_2 (1 + \frac{p_{k,i}^{\textnormal{OMA}} |\mathbf{h}_{k,i}^{H} \mathbf{w}_{k}|^2}{\textnormal{0.5} \sigma^2}),
\end{equation}
\begin{equation} \label{eq:  OMA rate at BS for user-j}
  R_{k,\textnormal{BS}}^{j,\textnormal{OMA}}= \textnormal{0.5} W  \log_2 (1 + \frac{p_{k,i}^{\textnormal{OMA}} |\mathbf{h}_{k,i}^{H} \mathbf{w}_{k}|^2}{\textnormal{0.5} \sigma^2}).
\end{equation}
\indent Also, the offloaded data processed by the \ac{MEC} server pool at the \ac{BS} for the $i$-th user and $j$-th user are given as
\begin{equation} \label{eq: OMA offloading data of user i}
  L_{k,i}^{\textnormal{off,OMA}}= t_{k,i:\textnormal{Phase2}}^{\textnormal{OMA}} R_{k,\textnormal{BS}}^{i,\textnormal{OMA}},
\end{equation}
\begin{equation} \label{eq: OMA offloading data of user j}
  L_{k,j}^{\textnormal{off,OMA}}= \textnormal{min} \{t_{k,j}^\textnormal{OMA} R_{k,i}^{j,\textnormal{OMA}}, t_{k,i:\textnormal{Phase3}}^{\textnormal{OMA}} R_{k,\textnormal{BS}}^{j,\textnormal{OMA}}\}.
\end{equation}
\indent Moreover, the energy consumption for data offloading of the $i$-th user and the $j$-th user can be expressed, respectively, as
\begin{equation} \label{eq: OMA energy of user i}
E_{k,i}^{\textnormal{off}}=(t_{k,i:\textnormal{Phase2}}^{\textnormal{OMA}}+t_{k,i:\textnormal{Phase3}}^{\textnormal{OMA}}) p_{k,i}^\textnormal{OMA},
\end{equation}
\begin{equation} \label{eq: OMA energy of user j}
 E_{k,j}^{\textnormal{off}}=t_{k,j}^\textnormal{OMA} p_{k,j}^\textnormal{OMA}.
\end{equation}
\indent By substituting~\eqref{eq:  OMA rate at user-i}-\eqref{eq: OMA energy of user j} into (P1) and following the same steps mentioned in the Appendix, one can find the optimal time and power allocation that minimize the total energy consumption of all the users in the \ac{OMA} counterpart scheme, as expressed in~\eqref{eq: OMA}. For a fair comparison between the proposed \ac{THz}-\ac{NOMA} scheme with full-offloading model and this baseline scheme we assume that $t_{k,i}=t_{k,i:\textnormal{Phase2}}^{\textnormal{OMA}}+t_{k,i:\textnormal{Phase3}}^{\textnormal{OMA}}$. 

\begin{gather} \label{eq: OMA}
  t_{k,j}^{*,\textnormal{OMA}} = \frac{T}{K}- \frac{L_{k,i}}{W \log_2 (1 +\frac{\beta_{k,i}}{\beta_{k,j}})} , \\ 
  t_{k,i:\textnormal{Phase2}}^{*,\textnormal{OMA}} = t_{k,i:\textnormal{Phase3}}^{*,\textnormal{OMA}} = \frac{ 0.5\ L_{k,i}}{W \log_2 (1 +\frac{\beta_{k,i}}{\beta_{k,j}})} , \\  
  p_{k,j}^{*,\textnormal{OMA}}=\frac{\sigma^2(2^{L_{k,j}/W t_{k,j}^{*,\textnormal{OMA}}}-1)}{|h_{k,ji}|^2} , \\ p_{k,i}^{*,\textnormal{OMA}}=\frac{0.5\ \sigma^2(2^{L_{k,j}/0.5\ W t_{k,i:\textnormal{Phase2}}^{*,\textnormal{OMA}}}-1)}{|\mathbf{h}_{k,i}^{H} \mathbf{w}_{k}|^2} .
\end{gather}

\indent Fig.~\ref{fig: benchmark with offloading models} illustrates the total energy consumption of the users, in Joules, for the proposed \ac{THz}-\ac{NOMA} system with the full-offloading model compared to its counterpart \ac{THz}-\ac{NOMA} system with the partial-offloading model as well as its counterpart system without offloading. For the partial-offloading model, we assume that $80$\% of the task input-bits for each user are offloaded and $20$\% of the task input-bits are locally computed. The first observation here for the three systems (i.e., Fig.~\ref{fig: Full offloading} - Fig.~\ref{fig: Without offloading}), is that as the number of users in these systems increases the total energy consumption increases. Likewise, as the number of task input-bits for each user increases, a somehow exponential increase in the total energy consumption of the users in these systems is observed; this is referred to the reason that with increasing the number of task input-bits, each user needs more power to transmit (offload) its bits to the \ac{BS} for remote computation. The second observation here is that for the systems with partial offloading and without offloading models, the users' energy consumption increases dramatically compared to the system with the full-offloading model, as a tremendous amount of energy is required to locally compute the large number of task input-bits in each user. In reality, this required energy might be unfeasible for the limited-battery mobile devices. Hence, our proposed \ac{THz}-\ac{NOMA} system with the full-offloading model has significant energy reduction for the systems that need several Gbits offloading data to process.

\begin{figure}[!t]
    \centering
   \vspace*{-.08in}
    \subfloat[Full-offloading model.]{
        \includegraphics[scale=0.485]{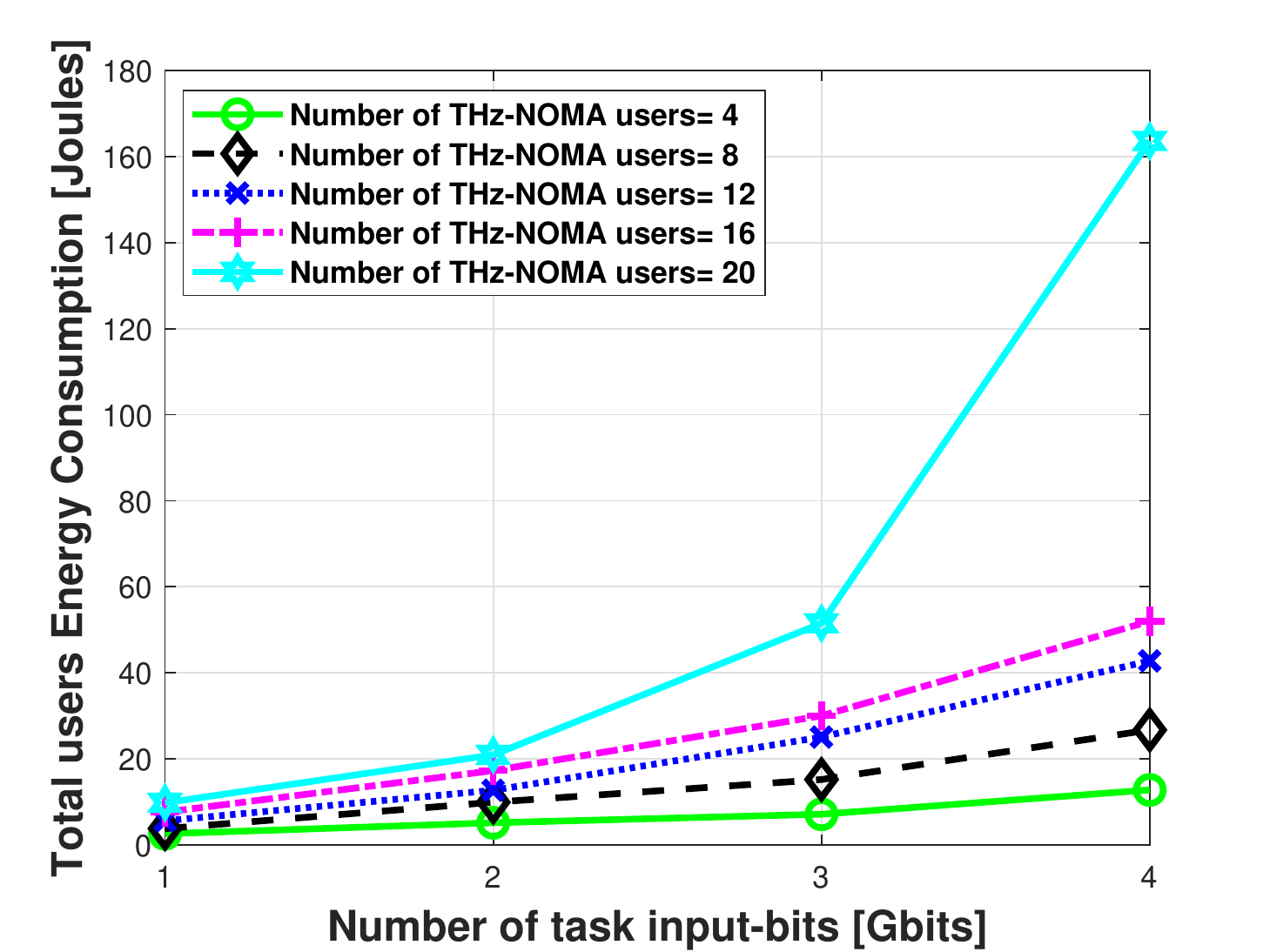}
        \label{fig: Full offloading}
    }
    \hfill
    \subfloat[Partial-offloading model.]{
       \includegraphics[scale=0.485]{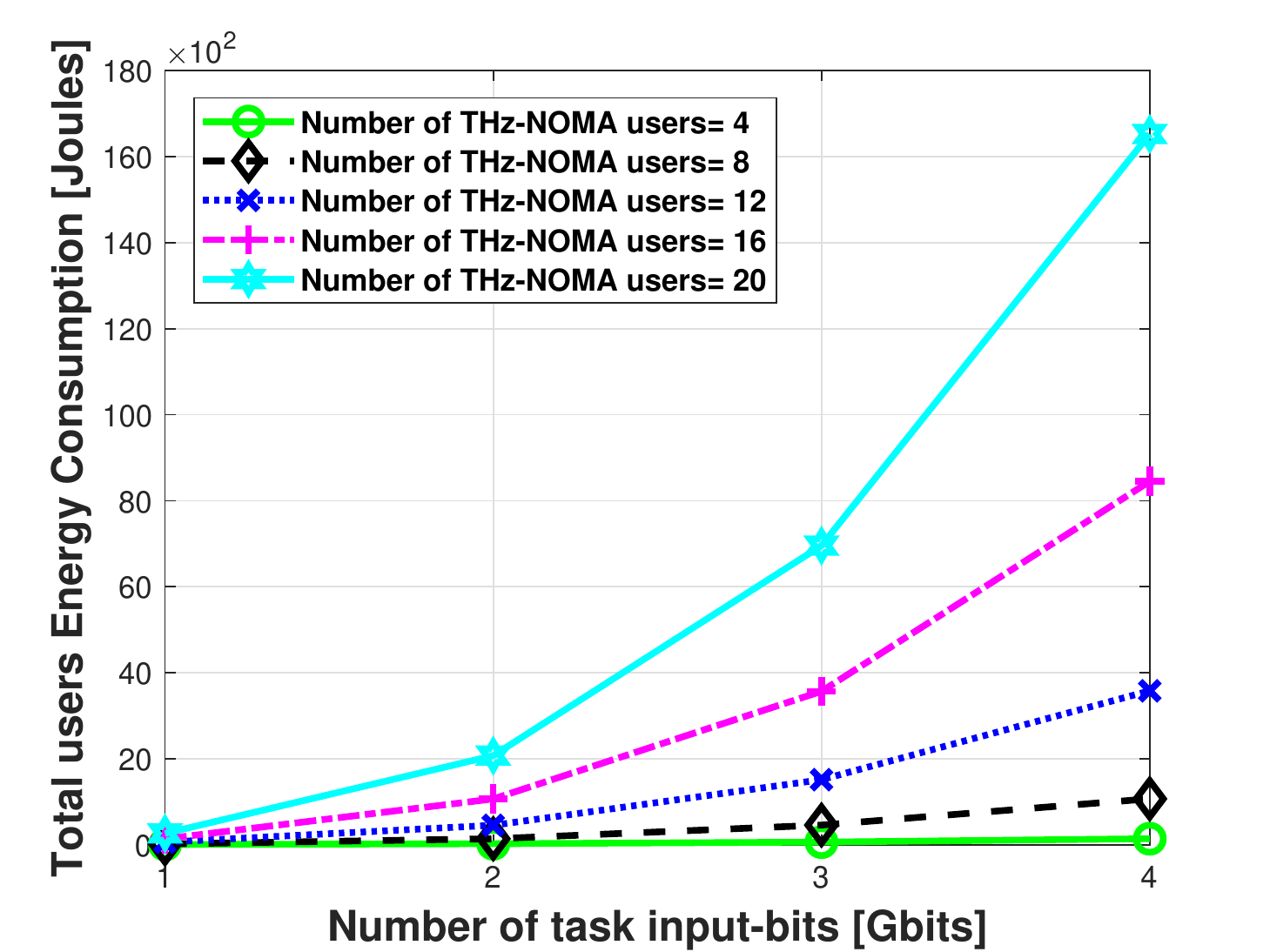}
        \label{fig: Partial offloading}
    }
    \hfill
   \subfloat[Without-offloading.]{
        \includegraphics[scale=0.485]{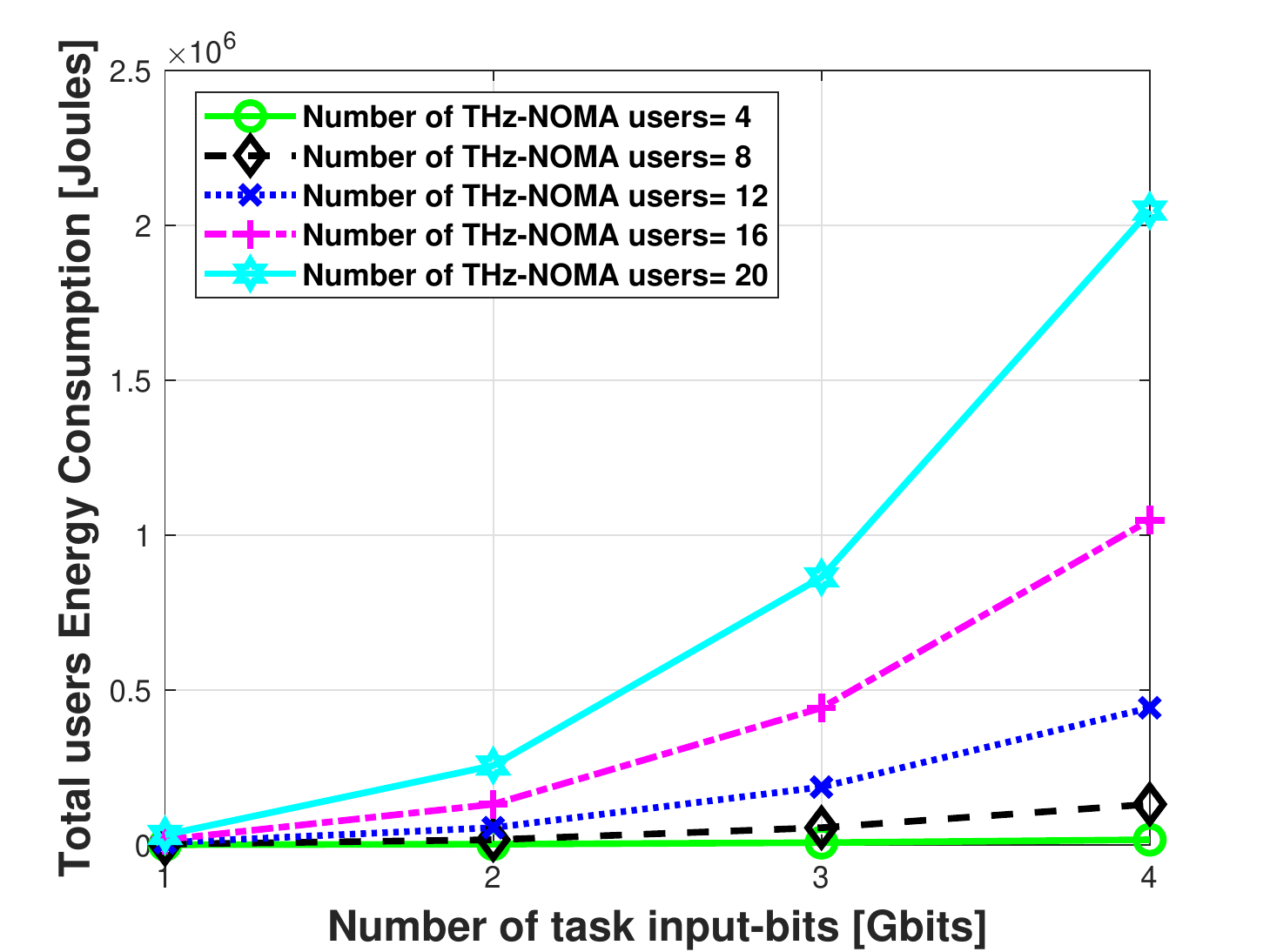}
       \label{fig: Without offloading}
    }
    \caption{The total users' energy consumption of the proposed \ac{THz}-\ac{NOMA} system for different offloading models.}
\label{fig: benchmark with offloading models}
\end{figure}

\begin{figure}[!t]
    \centering
   \vspace*{-.08in}
    \subfloat[]{
        \includegraphics[scale=0.485]{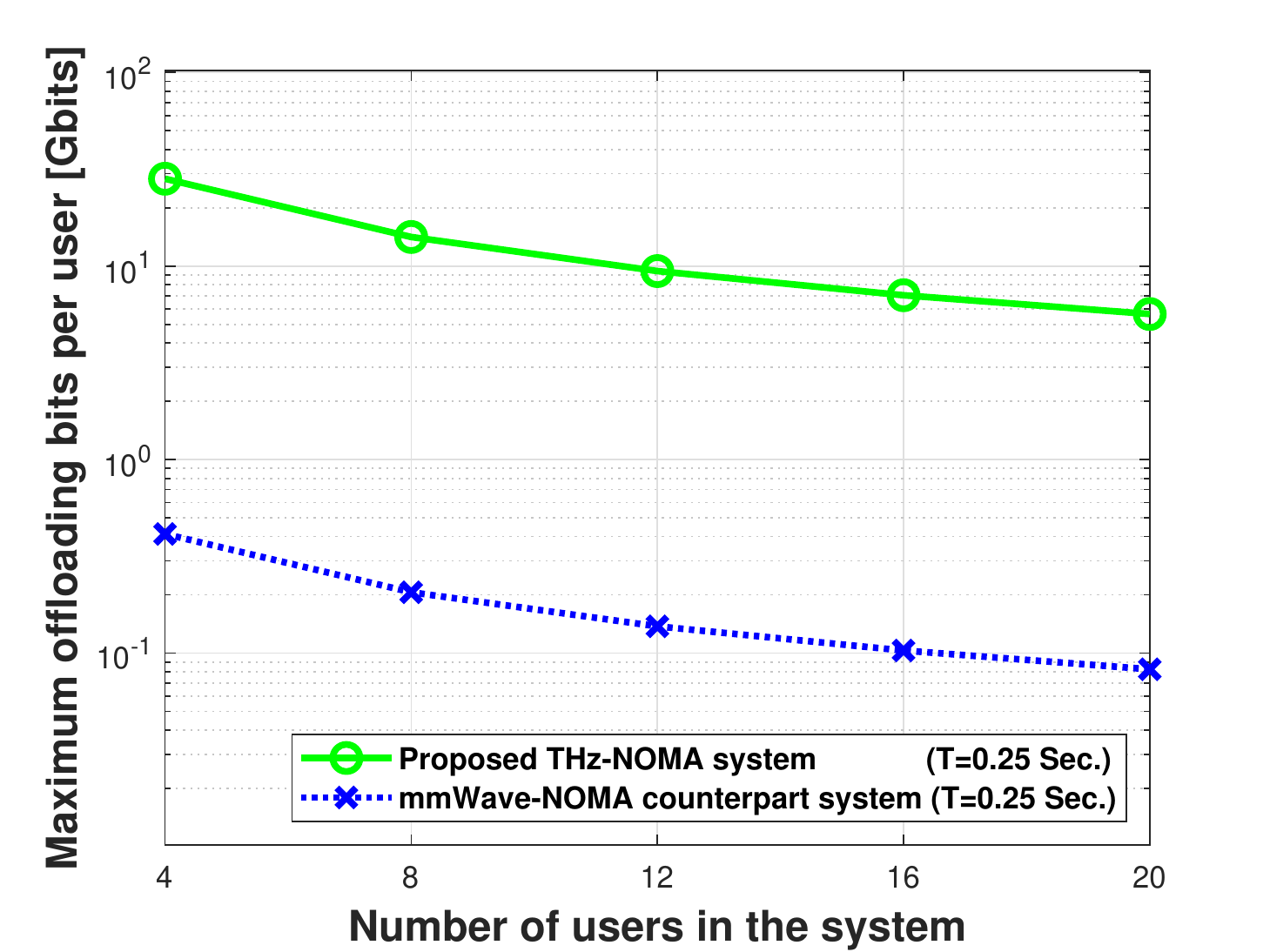} 
        \label{fig: Maximum offloading bits}
    }
    \hfill
    \subfloat[]{
       \includegraphics[scale=0.485]{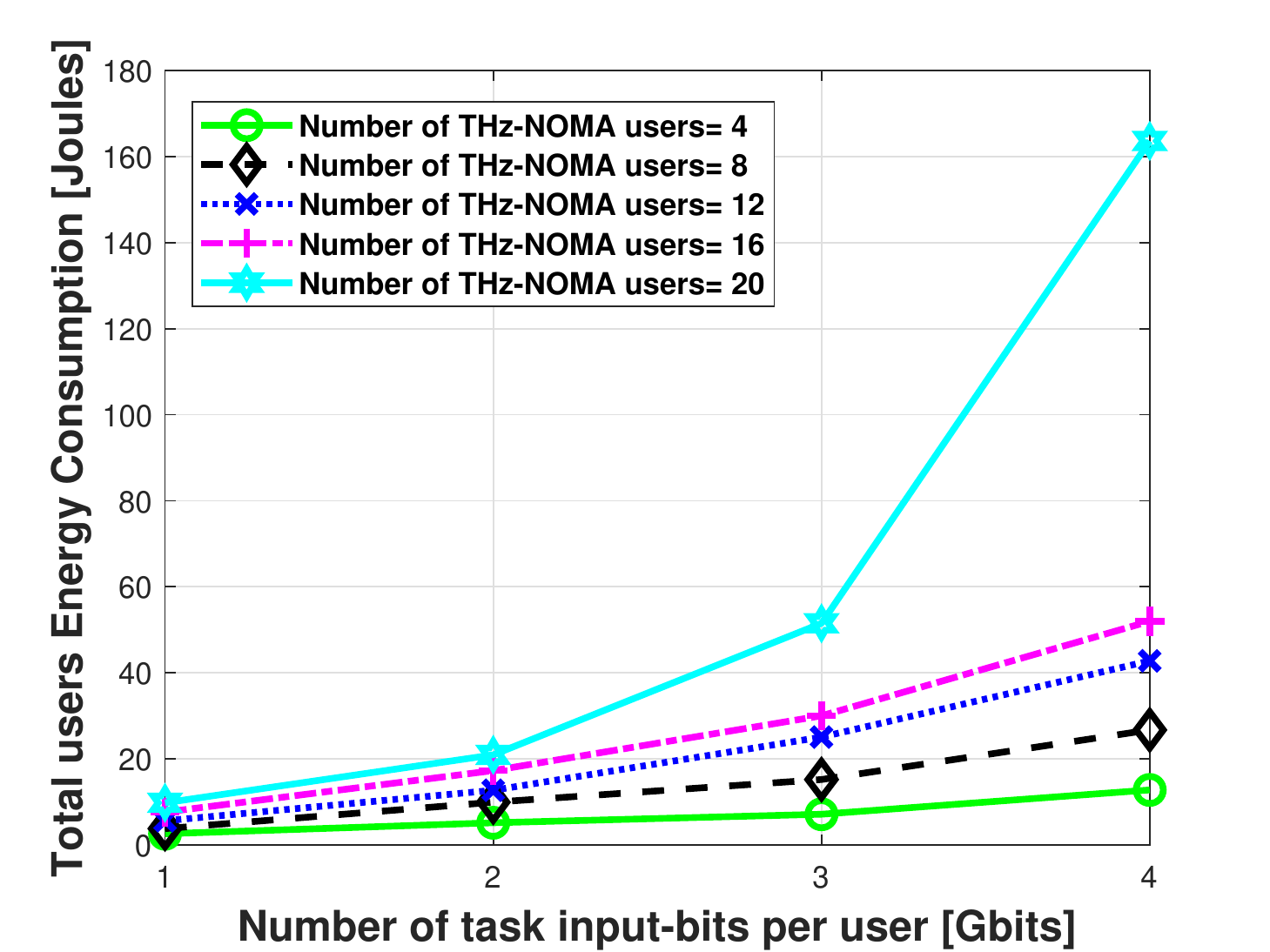} 
        \label{fig: Bits large}
    }
    \hfill
   \subfloat[]{
        \includegraphics[scale=0.485]{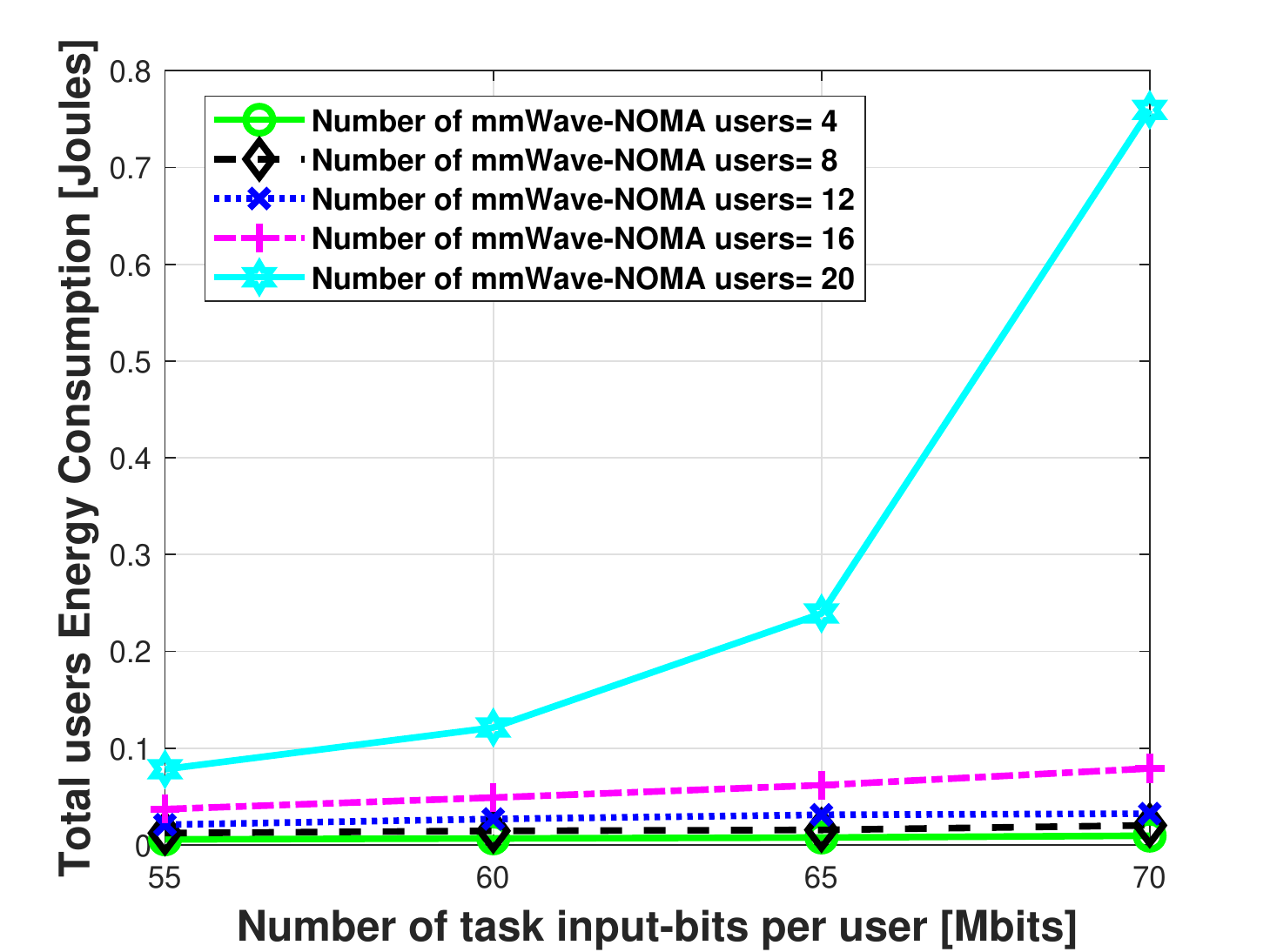} 
       \label{fig: Bits medium}
    }
    \caption{(a) The maximum offloading bits per user in the proposed \ac{THz}-\ac{NOMA} system as well as in its \ac{mmWave}-\ac{NOMA} counterpart. The total energy consumption of the users for different systems, (b) proposed \ac{THz}-\ac{NOMA} system, (c) \ac{mmWave}-\ac{NOMA} counterpart system.} 
\label{fig: Bits}
\end{figure}

In Fig.~\ref{fig: Bits}, we present the maximum offloading bits per user, in Gbits, as well as the total energy consumption of the users for the proposed \ac{THz}-\ac{NOMA} system compared to its \ac{mmWave}-\ac{NOMA} counterpart system. From Fig.~\ref{fig: Maximum offloading bits}, one can see that for the same latency constraint (i.e., $T=0.25$ second), the proposed \ac{THz}-\ac{NOMA} system can offload a considerably larger amount of data bits from each user to the \ac{BS}. We include the \ac{mmWave}-\ac{NOMA} counterpart scheme to answer the following question: Is it necessary to go to \ac{THz} band to serve users with several Gbits of raw data? To answer this question, by comparing Fig.~\ref{fig: Bits large} to Fig.~\ref{fig: Bits medium}, we can see that despite that the total users' energy consumption of \ac{mmWave}-\ac{NOMA} counterpart system is much smaller compared to the proposed \ac{THz}-\ac{NOMA} system, but the \ac{mmWave}-\ac{NOMA} counterpart system is only capable of offering tens of Mbits of data for each user to offload, while the proposed \ac{THz}-\ac{NOMA} system can offer several Gbits of data for each user to offload.

In Fig.~\ref{fig: NOMAvsOMA}, the total energy consumption of all users is provided for the proposed \ac{THz}-\ac{NOMA} system compared to its \ac{THz}-\ac{OMA} counterpart. In this figure, we notice that the total users' energy consumption of the proposed \ac{THz}-\ac{NOMA} system is slightly smaller, a few Joules, than the total users' energy consumption of the \ac{THz}-\ac{OMA} counterpart. This is because, with the assumption $t_{k,i}=t_{k,i:\textnormal{Phase2}}+t_{k,i:\textnormal{Phase3}}$, the time allocation of the \ac{THz}-\ac{NOMA} system and its \ac{THz}-\ac{OMA} counterpart are the same. Also, there is no change in the values of $p_{k,j}, \ \forall k \in [1, ..., K]$. The only change occurs in the values of $p_{k,i}, \ \forall k \in [1, ..., K]$. These $p_{k,i}$ values approximately increase by $7.8$x and this increase, in turn, increases the total users' energy consumption of the \ac{THz}-\ac{OMA} counterpart system slightly. Keeping in mind that the \ac{THz}-\ac{NOMA} system is more resource-efficient than the \ac{THz}-\ac{OMA} system, as one transmission is made towards the \ac{BS} compared to two transmissions for the \ac{THz}-\ac{OMA} system. 

\begin{figure}
    \centering
    \begin{minipage}{0.485\textwidth}
        \centering
        \includegraphics[scale=0.485]{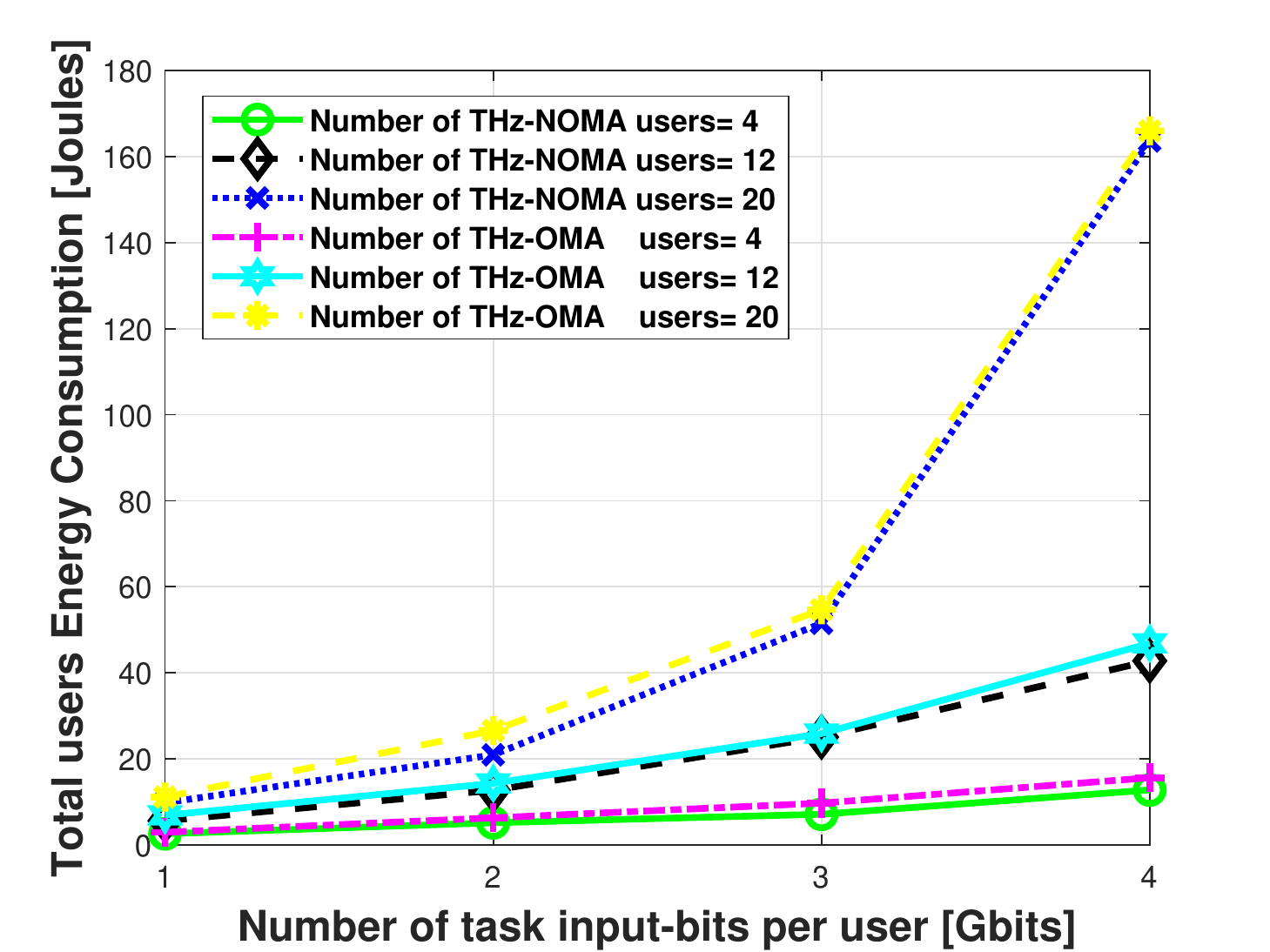} 
        \vspace{-0.5em}
        \caption{The total energy consumption of all users for the proposed \ac{THz}-\ac{NOMA} system compared to its \ac{THz}-\ac{OMA} counterpart.}
        \label{fig: NOMAvsOMA}
    \end{minipage}\hfill
    \begin{minipage}{0.485\textwidth}
        \centering
        \includegraphics[scale=0.54]{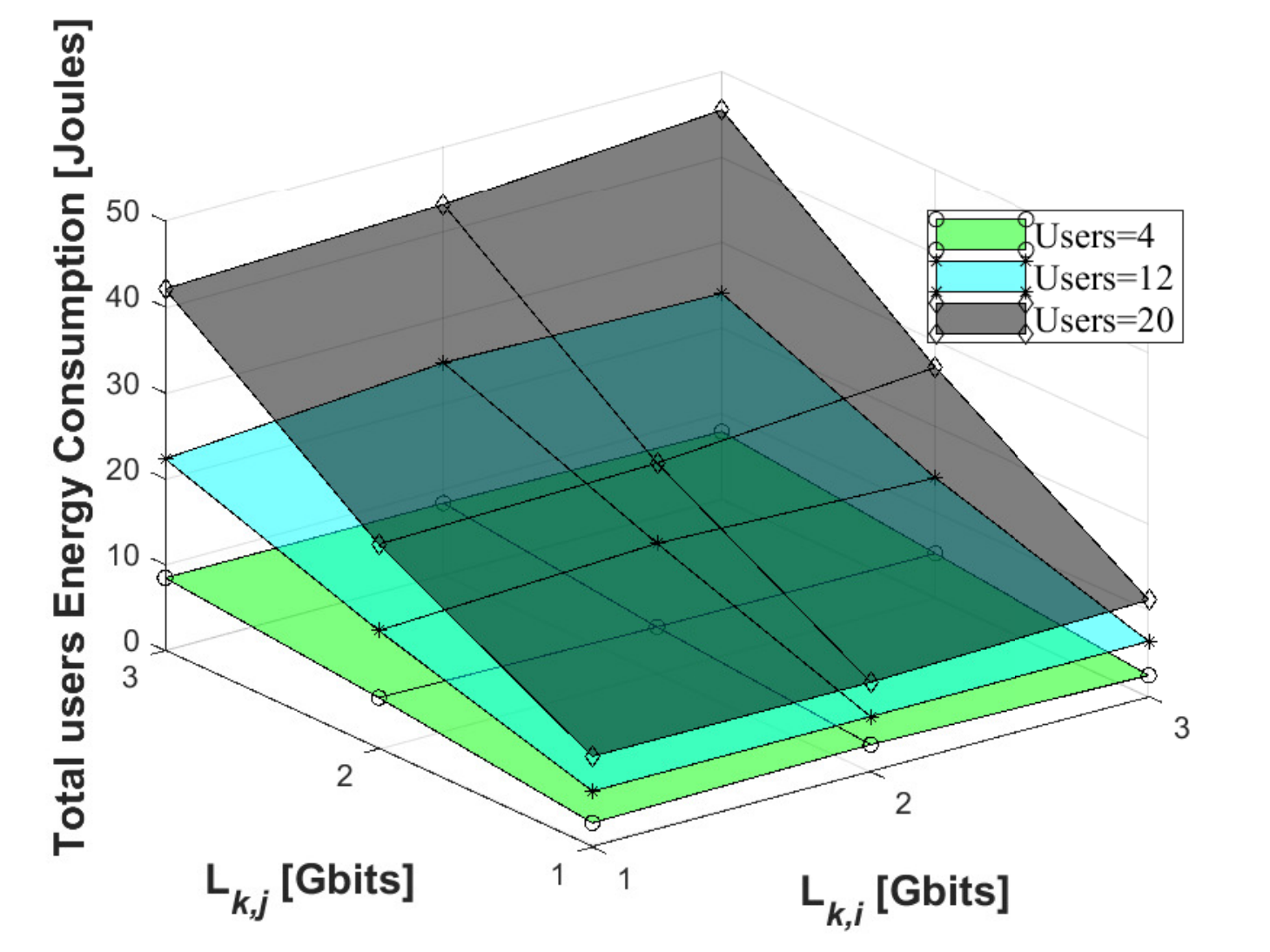} 
        \vspace{-0.5em}
        \caption{The offloading bits for the cell-center and cell-edge users in each \ac{NOMA} user-pair (i.e., $[L_{k,i},L_{k,j}], \ \forall k \in [1, ..., K]$) v.s. energy.}
        \label{fig: Changing_L_Energy}
    \end{minipage}
\end{figure}

Fig.~\ref{fig: Changing_L_Energy} shows the total users' energy consumption performance of the proposed \ac{THz}-\ac{NOMA} system while considering a variable offloading bits for the cell-edge users and the cell-center users in each \ac{NOMA} user-pair. In this figure, the increase in the offloading data of the cell-center users (i.e., $L_{k,i}, \ \forall k \in [1, ..., K]$) does not affect the total users' energy consumption much as compared to the effect of increasing the offloading data of the cell-edge users (i.e., $L_{k,j}, \ \forall k \in [1, ..., K]$). This is referred to the reason that the offloading bits of each cell-edge user is transmitted twice, (i.e., the first time is from each cell-edge user to each cell-center user, and the second time is from each cell-center user to the \ac{BS}) compared to the offloading bits of each cell-center user which is transmitted only once (i.e., from each cell-center user to the \ac{BS}). The increase in the total users' energy consumption for both cases of increasing $L_{k,i}$ or increasing $L_{k,j}$ is related to the reason that the required users' transmission powers needed to offload all these offloading bits add up and subsequently increase the total users' energy consumption in the system (recall that the total users' energy consumption is $E_{k,i}^{\textnormal{off}}+E_{k,j}^{\textnormal{off}}=t_{k,i} p_{k,i}+t_{k,j} p_{k,j}, \ \forall k \in [1, ..., K]$).

\begin{figure}[!t]
   \centering
   \vspace*{-.02in}
    \subfloat[]{
        \hspace*{-.15in}
        \includegraphics[scale=0.485]{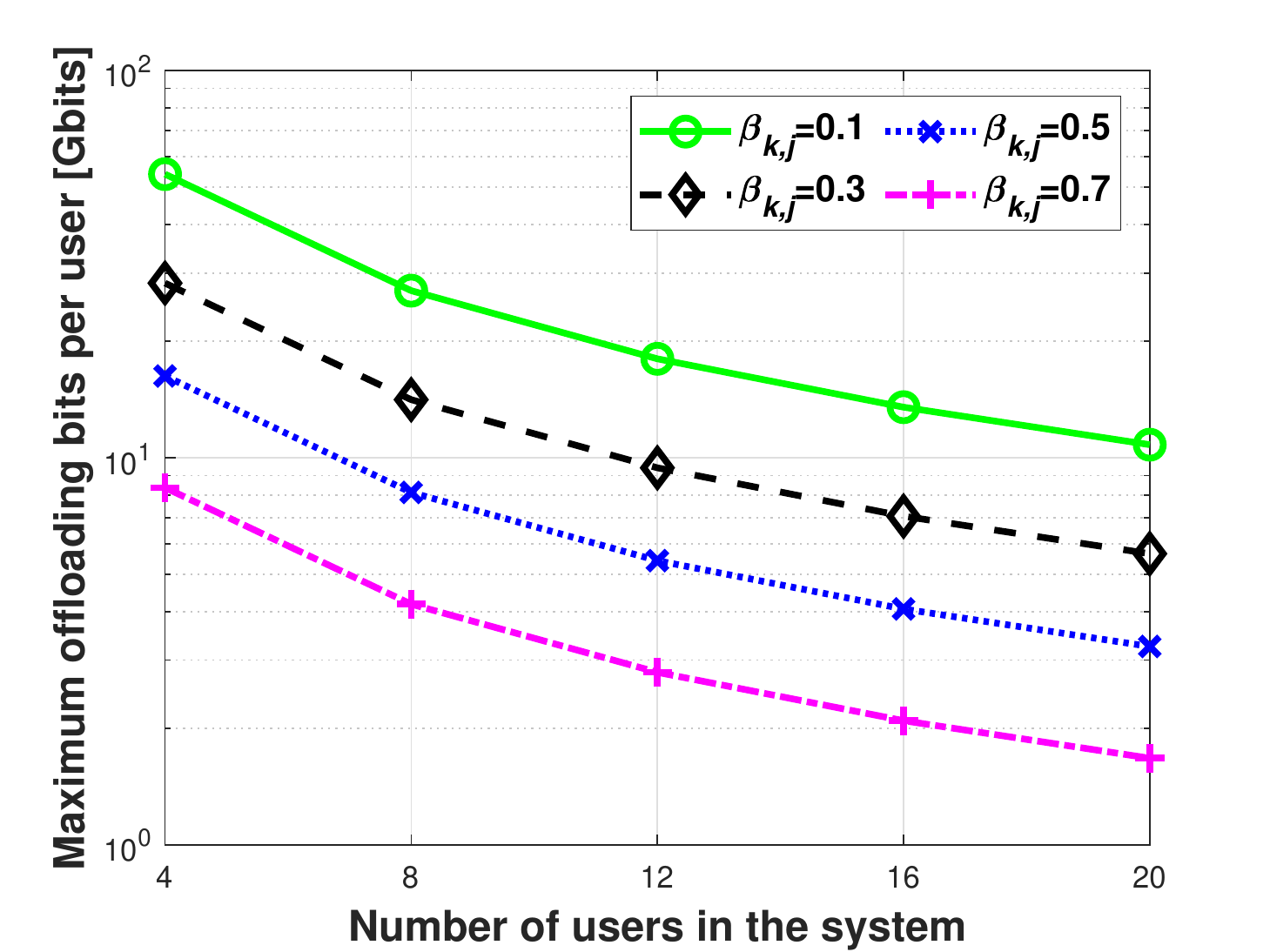}
        \label{fig: Changing_betaj_offloading_bits}
    }
    \hfill
    \subfloat[]{
        \hspace*{-.15in}
       \includegraphics[scale=0.485]{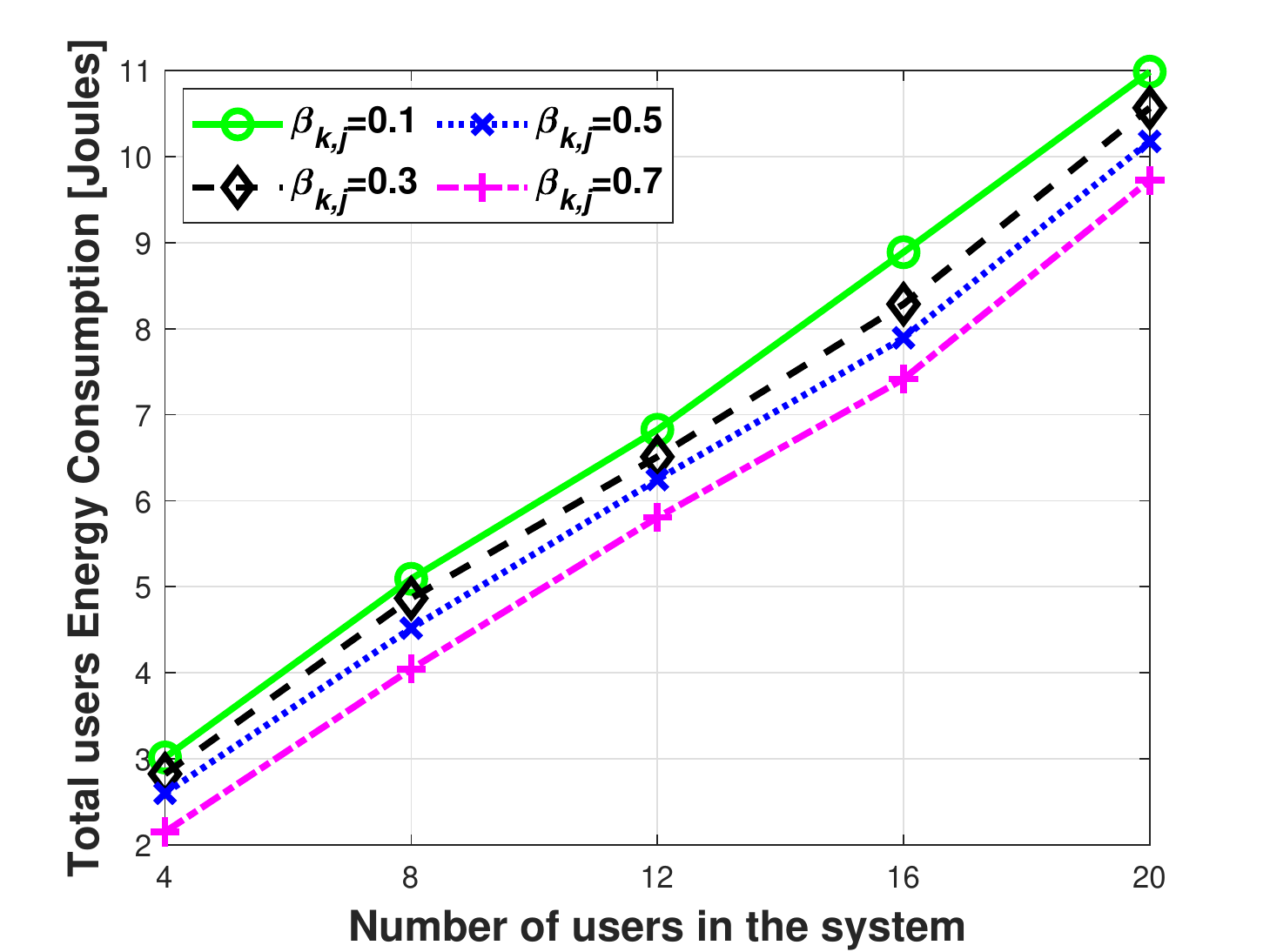}
        \label{fig: Changing_betaj_Energy}
    }
    \caption{The performance of the proposed \ac{THz}-\ac{NOMA} system while changing $\beta_{k,j}$ for (a) the maximum offloading bits per user, and (b) the total energy consumption of all users.}  
\label{fig: Changing_betaj}
\end{figure}

\begin{table*}[!t]
\centering
\caption{POTENTIAL THz CHANNEL WINDOWS}
\label{Table: Potential THz channel windows}
\resizebox{\textwidth}{!}{%
\begin{tabular}{|l|c|c|c|c|c|c|c|c|c|}
\hline
& f1 & f2 & f3 & f4 & f5 & f6 & f7 & f8 & f9 \\ \hline
Center frequency (THz)~\cite{singh2020analytical}                    & 1.51 & 2.52 & 3.42 & 4.91 & 5.72 & 6.57 & 7.19 & 8.83 & 9.57 \\ \hline
Contiguous bandwidth (GHz)~\cite{singh2020analytical}                & 169 & 82 & 137 & 113 & 126 & 120 & 246 & 217 & 230 \\ \hline
Absorption coefficient ($\textnormal{m}^{-1}$)~\cite{gordon2017hitran2016} & 0.1432 & 0.48 & 0.28 & 0.32 & 0.32 & 0.34 & 0.1344 & 0.1033 & 0.0779 \\ \hline
\end{tabular}%
}
\end{table*}

Fig.~\ref{fig: Changing_betaj} presents the maximum offloading bits per user as well as the total energy consumption of all users for the proposed \ac{THz}-\ac{NOMA} system while changing the value of the \ac{NOMA} power fractions $\beta_{k,j}$. In Fig.~\ref{fig: Changing_betaj_offloading_bits}, the maximum offloading bits per user in the system decreases as the number of served users in the system increase and the value of $\beta_{k,j}$ increases. Such a decrease can be deduced by looking at~\eqref{eq: approx. partial diff. for lambda3} in the Appendix, where the increase in the number of users in the system is reflected by a decrease in the maximum transmission (offloading) time for each user (Recall that $t_{k,j}+t_{k,i}= T/K$). On the other hand, when the value of $\beta_{k,j}$ increases, the argument of the logarithm, in~\eqref{eq: approx. partial diff. for lambda3}, decreases which consequently decreases the maximum offloading bits per user. Fig.~\ref{fig: Changing_betaj_Energy} is obtained while assuming that the required offloading bits for cell-edge users and cell-center users are equal (i.e., $L_{k,j}$=$L_{k,i}$= $1$ Gbits). In Fig.~\ref{fig: Changing_betaj_Energy}, as we increase the value of $\beta_{k,j}$ for a specific number of users, the total users' energy consumption decreases. This is since, based on~\eqref{eq: offloading data of user i} and~\eqref{eq: offloading data of user j}, as $\beta_{k,j}$ increases we allocate more time to the cell-center user, $t_{k,i}$, to offload the cell-center and cell-edge users' data, through the \ac{NOMA} scheme, with less required power, $p_{k,i}$, for that offloading process.

\begin{figure}[!t]
    \centering
   \vspace*{-.02in}
    \subfloat[]{
        \includegraphics[scale=0.485]{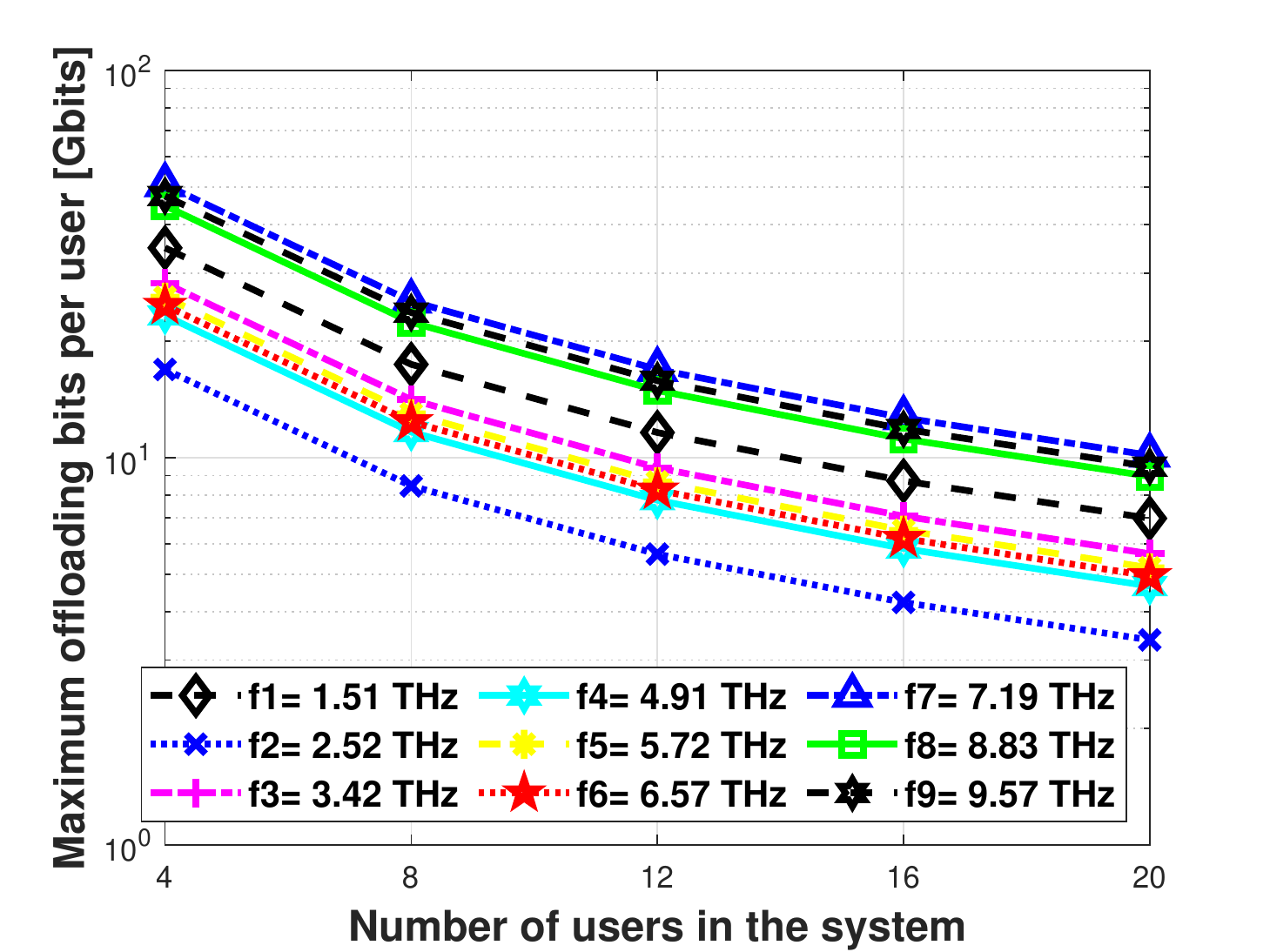}
        \label{fig: Changing_freqs_offloading_bits}
    }
    \hfill
    \subfloat[]{
       \includegraphics[scale=0.485]{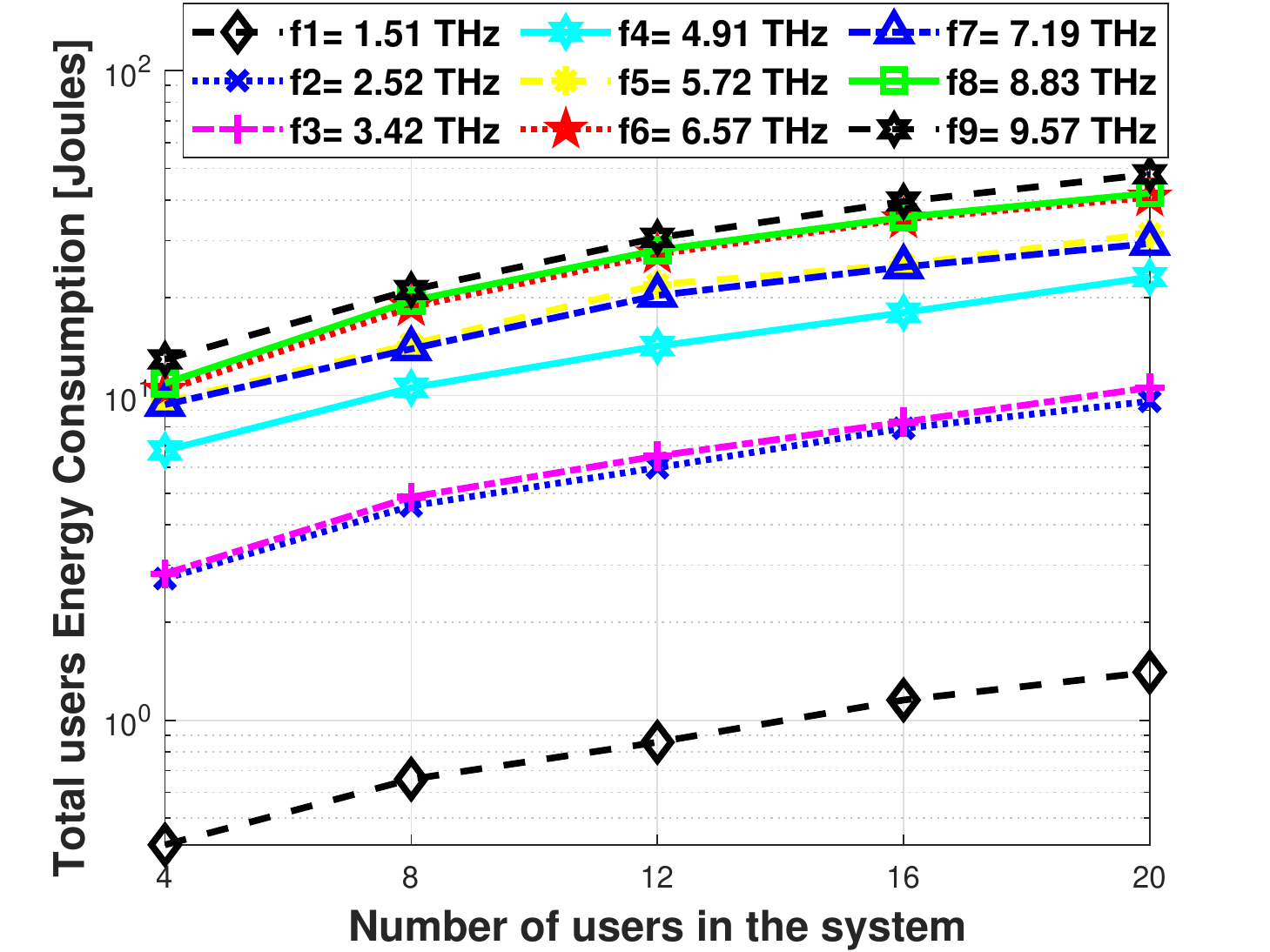}
        \label{fig: Changing_freqs_Energy}
    }
    \caption{The performance of the proposed \ac{THz}-\ac{NOMA} system while changing the \ac{THz} frequency windows for (a) the maximum offloading bits per user, and (b) the total energy consumption of all users.} 
\label{fig: Changing_freqs}
\end{figure}

In~\cite{singh2020analytical}, a sensitivity analysis for the \ac{THz} spectrum that ranges between $0.3-10.0$ \ac{THz} has been conducted. Based on this sensitivity analysis, the authors of~\cite{singh2020analytical} have identified some potential \ac{THz} channel windows that are $1.0$ \ac{THz} apart. In Table~\ref{Table: Potential THz channel windows}, a summary of these channel windows that includes its, (i) center frequencies, (ii) contiguous bandwidths, and (iii) absorption coefficients, is provided. Fig.~\ref{fig: Changing_freqs} provides the maximum offloading bits per user as well as the total energy consumption of all the users for the proposed \ac{THz}-\ac{NOMA} system for different \ac{THz} channel windows that were listed in Table~\ref{Table: Potential THz channel windows}. From Fig.~\ref{fig: Changing_freqs_offloading_bits}, one may observe that the maximum offloading bits per user is changing based on the contiguous bandwidth available in each \ac{THz} channel window. Fig.~\ref{fig: Changing_freqs_Energy} is also simulated while assuming that the required offloading bits for cell-edge users and cell-center users are equal (i.e., $L_{k,j}$=$L_{k,i}$= $1$ Gbits). In this figure, for the frequencies that are at the lower end of the \ac{THz} spectrum (i.e., f$1$-f$4$), we can see that as the center frequency decreases the total users' energy consumption decreases. On the other hand, for the upper end of the \ac{THz} spectrum (i.e., f$5$-f$9$), the effect of increasing the center frequency does not always decrease the total users' energy consumption in the system. For example, if we look at the curves of (f$5$-f$7$), despite that f$7$ has a larger center frequency than f$5$ and f$6$, its total users' energy consumption is lower than f$5$ and f$6$. This is because f$7$ has a smaller absorption loss, and consequently a smaller path loss, compared to f$5$ and f$6$ (See Table~\ref{Table: Potential THz channel windows}).

Fig.~\ref{fig: antennas} illustrates the total energy consumption of all the users while changing the number of antennas in the \ac{BS}. As apparent in the figure, the total energy consumption of all users decreases as the number of antennas in the \ac{BS} increases. This is since, based on~\eqref{eq: offloading data of user i} and~\eqref{eq: offloading data of user j}, as the number of antennas in the \ac{BS} increases the required transmission power for the cell-center users ($p_{k,i}, \ \forall k \in [1, ..., K]$) decrease while the values of $t_{k,j}, t_{k,i}, \textnormal{and} \ p_{k,j}, \ \forall k \in [1, ..., K]$ remain the same. 

\begin{figure}[!t]
\centering
\includegraphics[scale=0.4]{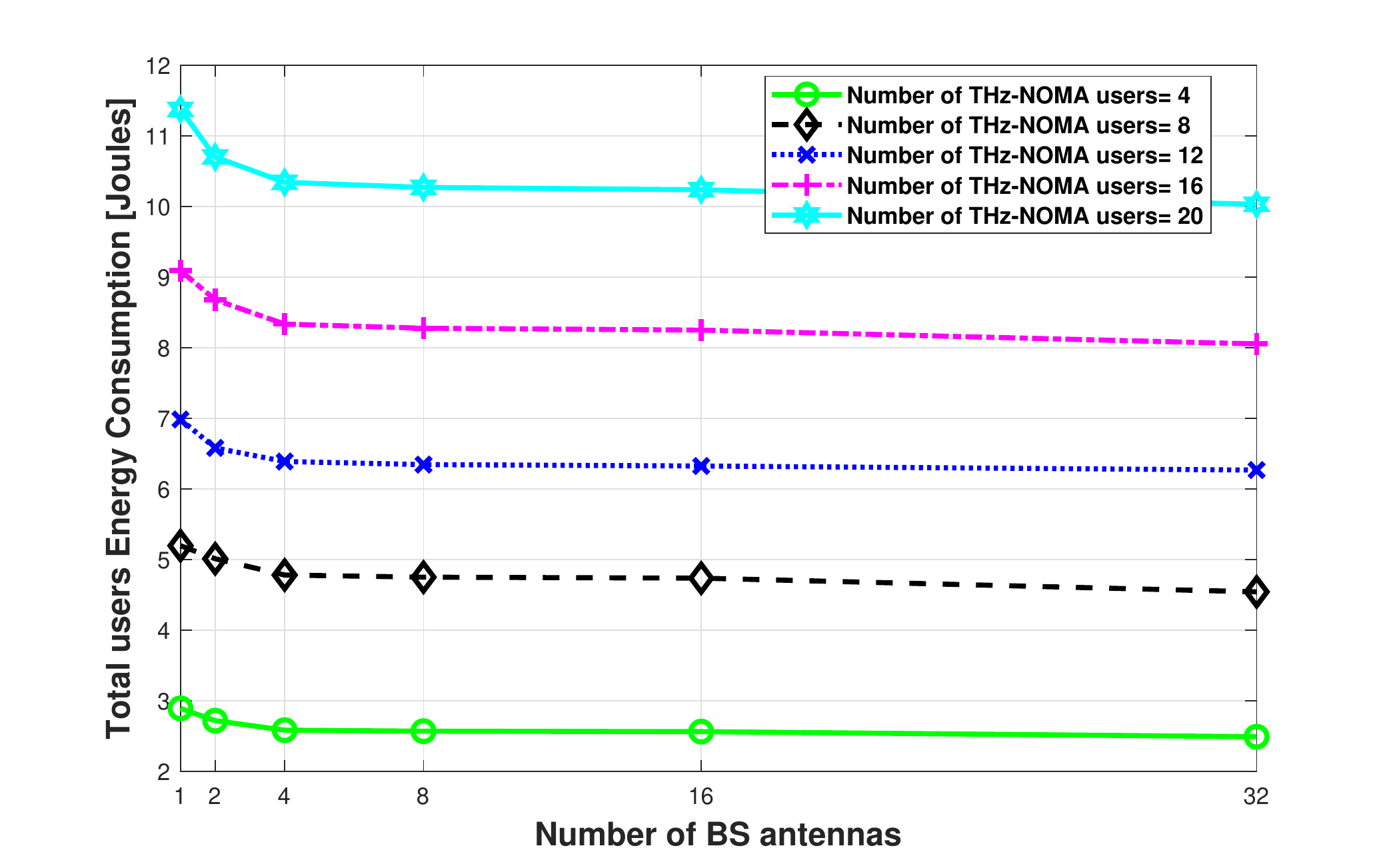}
\caption{The total energy consumption of all users for the proposed \ac{THz}-\ac{NOMA} system while changing the number of antennas at the \ac{BS}.}
\label{fig: antennas}
\end{figure}

\begin{figure}[!t]
    \centering
   \vspace*{-.02in}
    \subfloat[Full-offloading model.]{
        \includegraphics[scale=0.485]{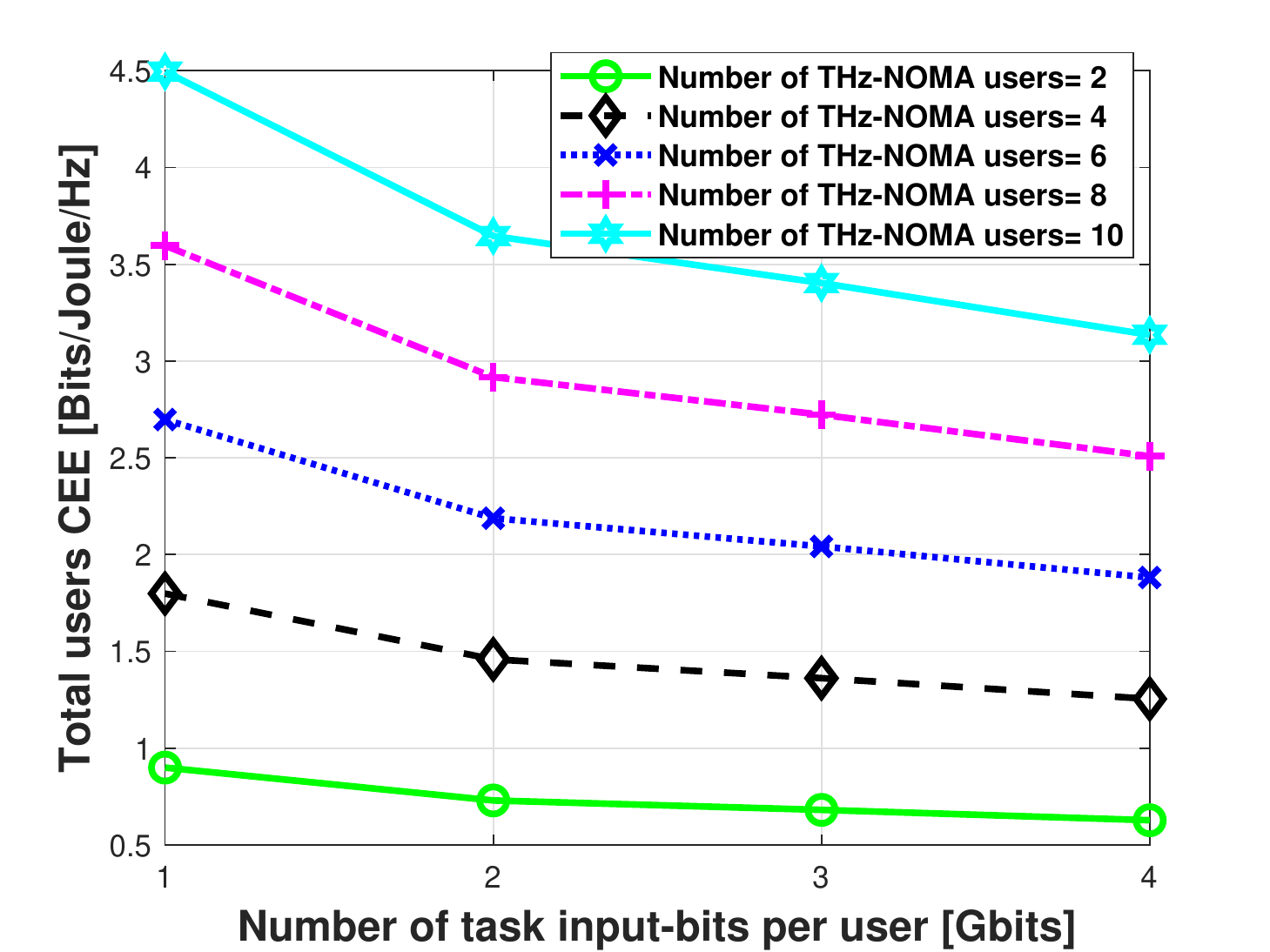}
        \label{fig: CEE Full offloading}
    }
    \hfill
    \subfloat[Partial-offloading model.]{
       \includegraphics[scale=0.485]{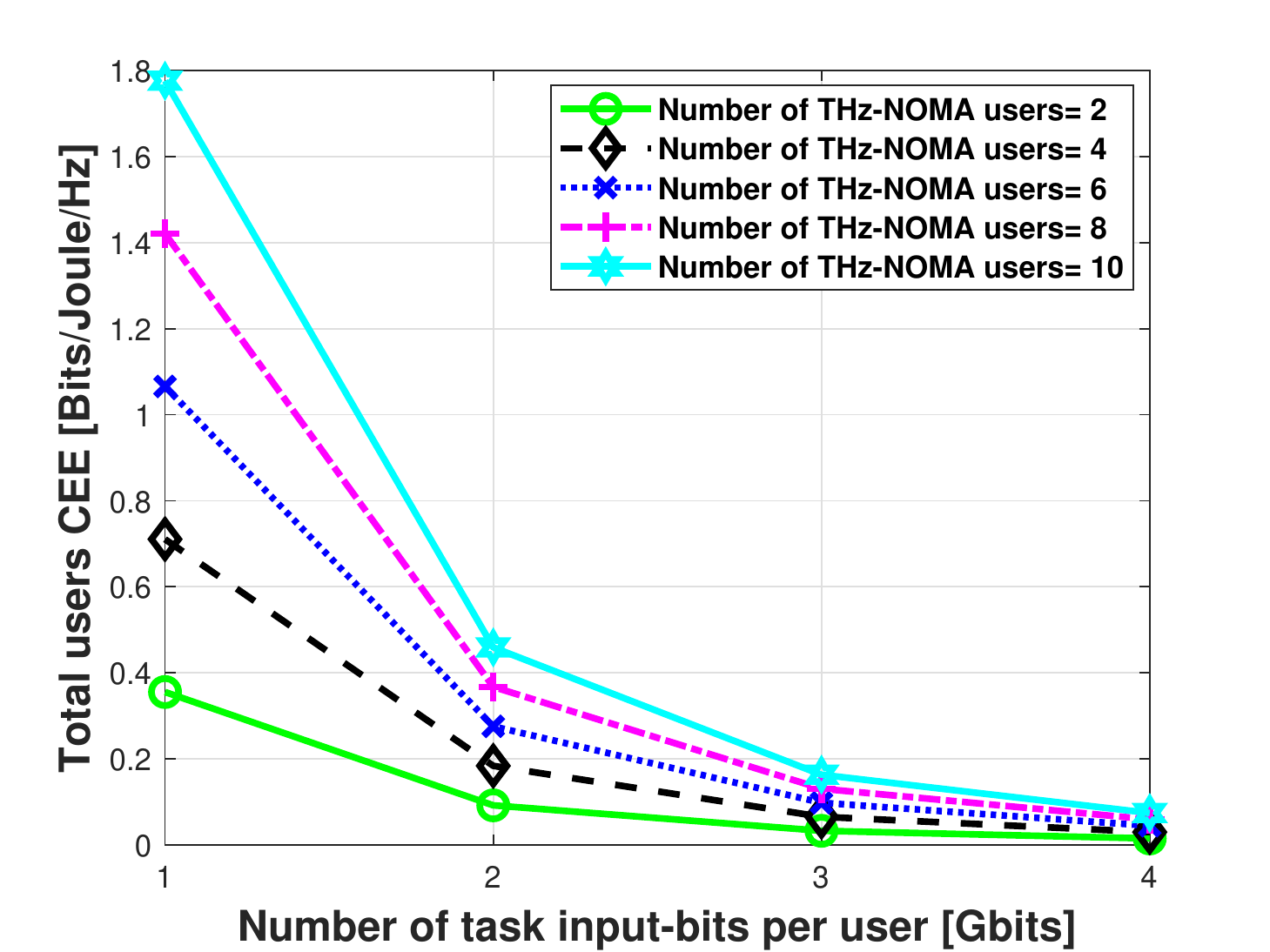}
        \label{fig: CEE Partial offloading}
    }
    \hfill
   \subfloat[Without-offloading.]{
        \includegraphics[scale=0.485]{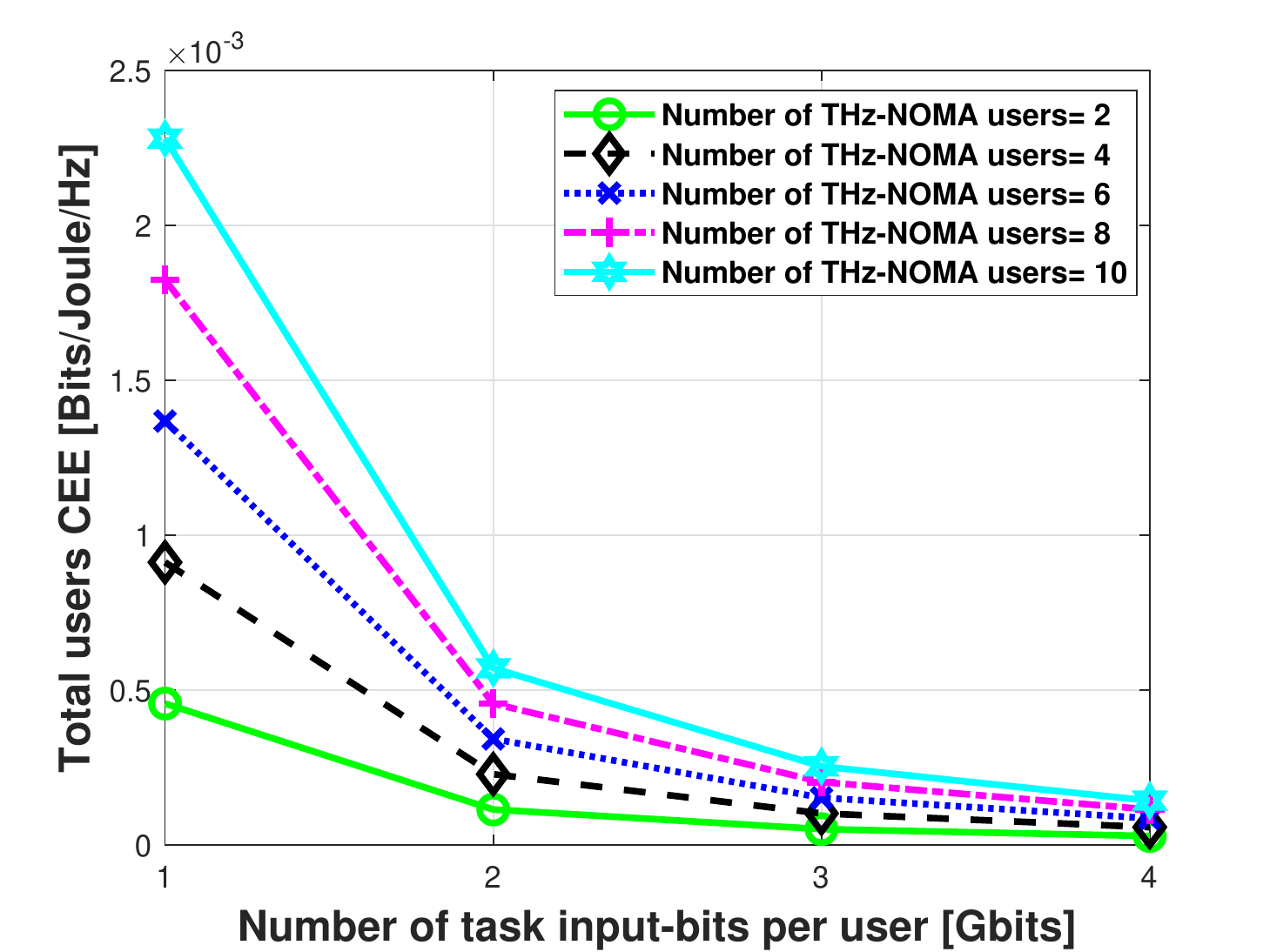}
       \label{fig: CEE Without offloading}
    }
    \caption{The total users' \ac{CEE} for the proposed \ac{THz}-\ac{NOMA} system for different offloading models.
    }
\label{fig: CEE benchmark with offloading models}
\end{figure}

Fig.~\ref{fig: CEE benchmark with offloading models} illustrates the total users' \ac{CEE}, in Bits/Joule/Hz, for the proposed \ac{THz}-\ac{NOMA} system with the full-offloading model compared to its counterpart \ac{THz}-\ac{NOMA} system with the partial-offloading model as well as the system without offloading. For the partial-offloading model, we also assume that $80$\% of the task input-bits for each user are offloaded and $20$\% of the task input-bits are locally computed. The first observation here for the three systems (i.e., Fig.~\ref{fig: CEE Full offloading} - Fig.~\ref{fig: CEE Without offloading}), is that as the number of task input-bits for each user increases, a somehow exponential decrease in the total users' \ac{CEE} in these systems is observed; this is referred to the reason that with increasing the number of task input-bits, each user needs more power to transmit (offload) its bits to the \ac{BS} for remote computation. Specifically, such an increase in the transmission power required for offloading users' task input-bits will increase the denominator of the \ac{CEE} formula (i.e.,~\eqref{eq: CEE formula}) while its numerator is fixed. The second observation here is that our proposed system with the full-offloading model provides better total users' \ac{CEE} compared to the systems with partial offloading and without offloading models. In consistency with the drawn conclusion of Fig.~\ref{fig: benchmark with offloading models}, a large amount of energy is required for the systems with partial offloading and without offloading models to locally compute the large number of task input-bits at each user.

\begin{figure}[!t]
    \centering
    \begin{minipage}{0.485\textwidth}
        \centering
        \includegraphics[scale=0.485]{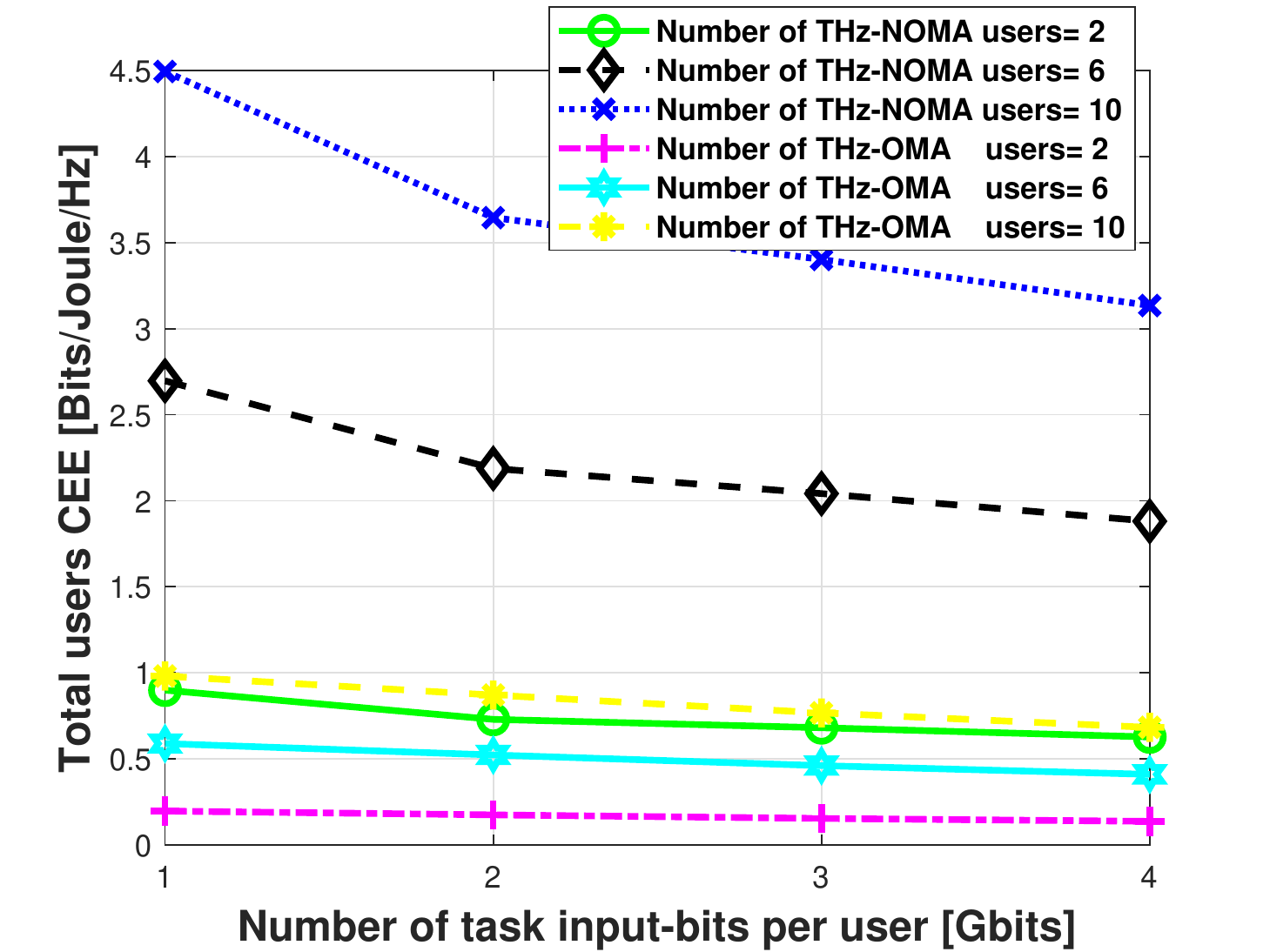} 
        \vspace{-0.5em}
        \caption{The total users' \ac{CEE} for the proposed \ac{THz}-\ac{NOMA} system compared to its \ac{THz}-\ac{OMA} counterpart.}
        \label{fig: CEE_compared_to_OMA}
    \end{minipage}\hfill
    \begin{minipage}{0.485\textwidth}
        \centering
        \includegraphics[scale=0.485]{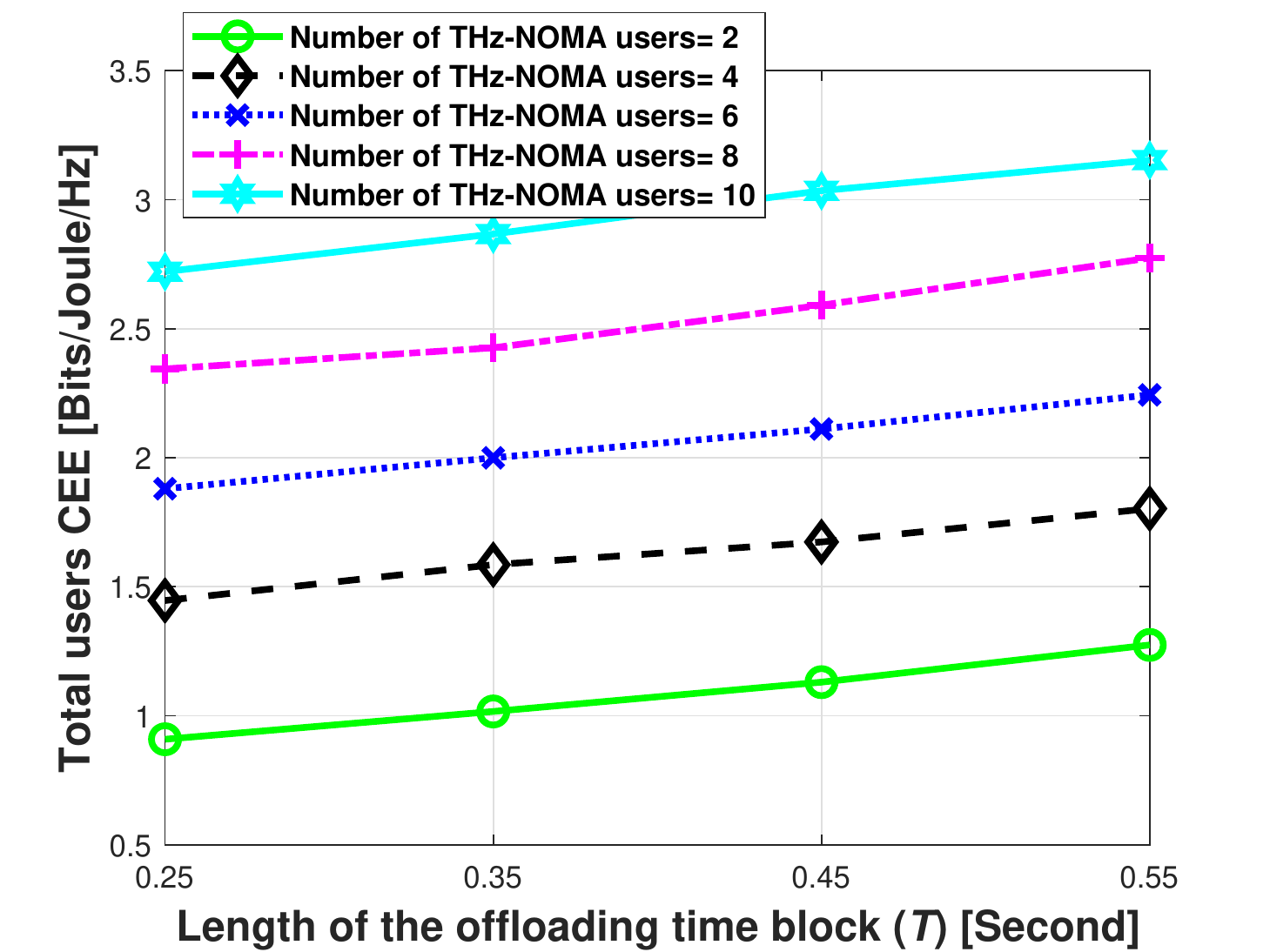} 
        \vspace{-0.5em}
        \caption{The total users' \ac{CEE} for the proposed \ac{THz}-\ac{NOMA} system while changing the length of the offloading time block.}
        \label{fig: CEE_T}
    \end{minipage}\hfill
    \begin{minipage}{0.485\textwidth}
        \centering
        \includegraphics[scale=0.485]{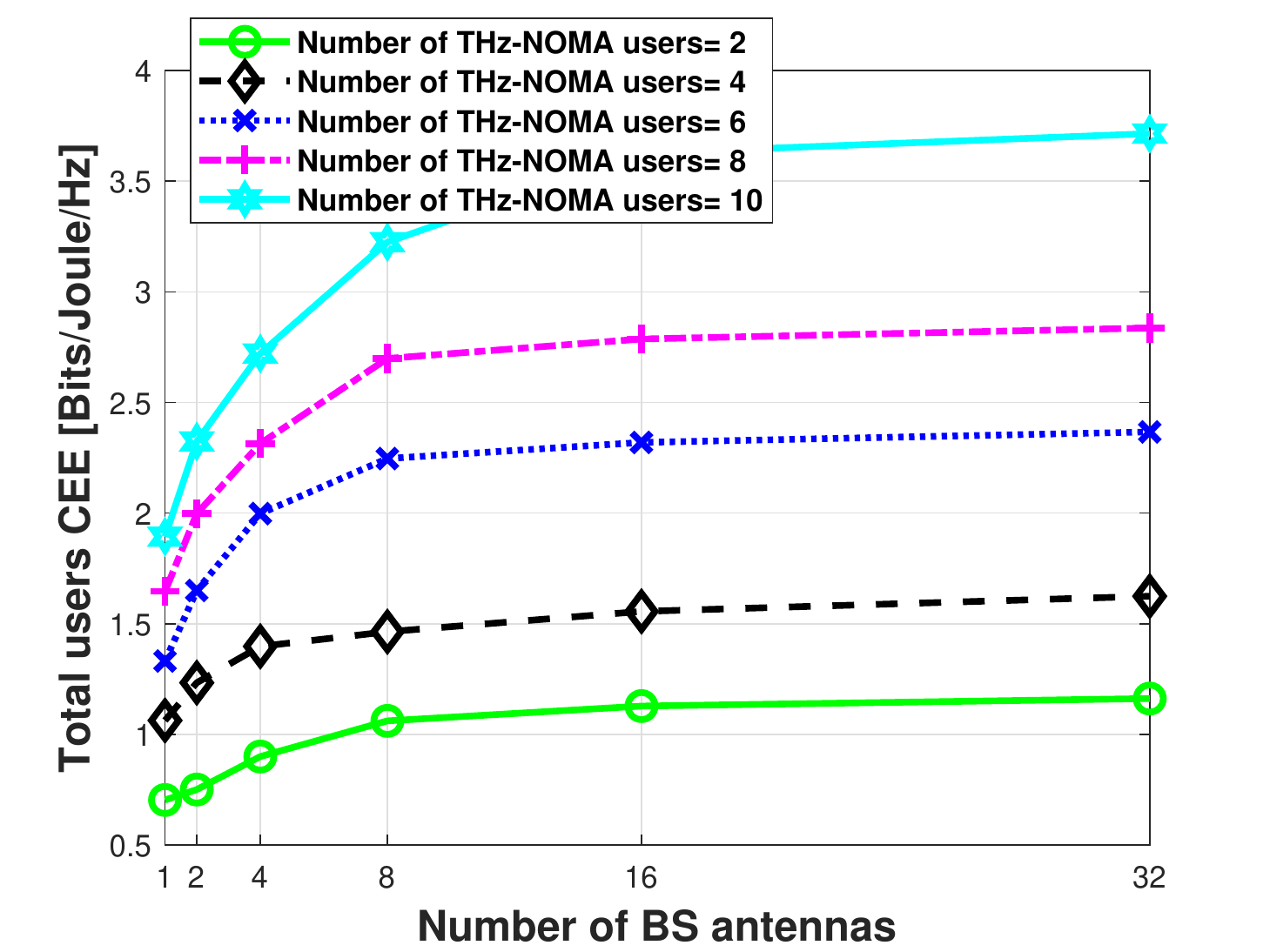} 
        \vspace{-0.5em}
        \caption{The total users' \ac{CEE} for the proposed \ac{THz}-\ac{NOMA} system while changing the number of antennas in the \ac{BS}.}
        \label{fig: CEE_antennas}
    \end{minipage}\hfill
\end{figure}

In Fig.~\ref{fig: CEE_compared_to_OMA}, the total users' \ac{CEE} is provided for the proposed \ac{THz}-\ac{NOMA} system compared to its \ac{THz}-\ac{OMA} counterpart. In this figure, one can see a significant increase in the total users' \ac{CEE} for the proposed \ac{THz}-\ac{NOMA} system compared to its \ac{THz}-\ac{OMA} counterpart. This increase is related to the reason that the proposed \ac{THz}-\ac{NOMA} system consumes less energy, as explained before in the discussions of Fig.~\ref{fig: NOMAvsOMA}, for the data offloading compared to its \ac{THz}-\ac{OMA} counterpart.

Fig.~\ref{fig: CEE_T} presents the total users' \ac{CEE} while changing the length of the offloading time block $T$. One can draw two insights from this figure, firstly, for a fixed offloading time block $T$ (e.g., $0.25$ second), as the number of users in the system increases the total users' \ac{CEE} increases. This is because the total users' \ac{CEE} is obtained from accumulating the \ac{CEE} of each \ac{NOMA} user-pair (e.g., with the number of \ac{THz}-\ac{NOMA} users equals $10$, the total number of the \ac{NOMA} user pairs is $5$). Secondly, the total users' \ac{CEE} increases monotonically as the value of the offloading time block $T$ increases. In~\eqref{eq: CEE formula}, despite the fact that the time block $T$ appears in both the numerator as well as the denominator (Recall that $t_{k,j}+t_{k,i}= T/K$) but while $T$ increases the required offloading power values, $p_{k,j}$ and $p_{k,i}$, decreases and hence the total users' \ac{CEE} increases. A similar trend can be found in Fig. 4 of~\cite{9345931}.

In Fig.~\ref{fig: CEE_antennas}, the total users' \ac{CEE} is provided while changing the number of antennas in the \ac{BS}. Clearly, increasing the number of receive antennas in the \ac{BS} increases the beam gain, $c_{k,i} = |\mathbf{h}_{k,i}^{H} \mathbf{w}_{k}|^2$,  in~\eqref{eq: CEE formula} and this improves the total users' \ac{CEE} performance.

\section{Conclusions}
\label{Section: Conclusions}
The growing interest in \ac{THz} communications as an enabler of high-data-rate low-latency applications in \ac{B5G} networks has spurred the research on the integration of various promising communication schemes. In this paper, a cooperative scheme for the \ac{NOMA}-assisted \ac{THz}-\ac{SIMO} \ac{MEC} system is proposed while reinforcing its performance by solving two energy-related optimization problems. This cooperative scheme allows the uplink transmission of offloaded data from the far cell-edge users to the \ac{BS} and comprises four stages: (i) a user pairing stage, (ii) a receive beamforming stage, (iii) a time allocation stage, and (iv) a \ac{NOMA} transmission power allocation stage. The proposed scheme is benchmarked with several baseline schemes and is shown to outperform them in terms of the successful handing of several Gbits/user offloaded data, with a reasonable total users' energy consumption within a predefined latency constraint. Then, the obtained results of the proposed scheme quantified the total users' energy consumption for various system parameters, such as (i) the number of served users, (ii) the number of task input-bits, (iii) the \ac{NOMA} power fractions allocated to transmit the users offloaded data bits, (iv) the variable input-bits for cell-edge and cell-center users, (v) the different potential \ac{THz} channel windows, and (vi) the number of \ac{BS} antennas. After that, additional results are presented to quantify the total users' \ac{CEE} for various system parameters as well. One of the possible interesting directions is to extend the proposed scheme to indoor multi-cell scenarios while utilizing some coordinated multi-point (CoMP) schemes to manage the inter-cell interference (ICI) at the cell-edge users.

\appendix[Proof of Proposition~\ref{proposition-1}]\label{appendix-1}

Here, the optimal solution of (P1) is provided through the Lagrange duality method. Firstly, (P1) needs to be reformulated (i) by replacing the optimization variables $p_{k,j}$ and $p_{k,i}$ with the variable vector $\mathbf{E}=\{E_{k,j}^{\textnormal{off}}=t_{k,j}p_{k,j},E_{k,i}^{\textnormal{off}}=t_{k,i}p_{k,i}\}$, and (ii) by replacing $t_{k,j}$ and $t_{k,i}$ with the variable vector $\mathbf{t}=\{t_{k,j},t_{k,i}\}$, and (iii) by substituting~\eqref{eq: offloading data of user i} and~\eqref{eq: offloading data of user j} into (P1). The resultant reformulated problem can be expressed as
\begin{alignat}{3}
\textnormal{(P2): }&\underset{\{\mathbf{t},\mathbf{E}\}}{\text{min}} & \ &E_{k,j}^{\textnormal{off}}+E_{k,i}^{\textnormal{off}}, \ k \in [1, ..., K] \label{eq:objective2}\\
& \ \text{s.t.} &  &  t_{k,j} R_{k,i}^j \geq L_{k,j},  \ k \in [1, ..., K] \label{eq:constraint-b2}\\
&  &  & t_{k,i} R_{k,\textnormal{BS}}^{j} \geq L_{k,j},  \ k \in [1, ..., K] \label{eq:constraint-c2} \\
&  &  & t_{k,i} R_{k,\textnormal{BS}}^{i} \geq L_{k,i},  \ k \in [1, ..., K] \label{eq:constraint-d2} \\
&  &  & t_{k,j}+t_{k,i} \leq T/K,  \ k \in [1, ..., K] \label{eq:constraint-e2}
\end{alignat}
\noindent where (P2) is jointly convex with respect to $\mathbf{E}$ and $\mathbf{t}$. The convexity of (P2) can be proved in a similar way to the proof in the Appendix of~\cite{9348649}. Secondly, as (P2) is convex, the Lagrangian dual function (P3) can be expressed, while assigning $\boldsymbol{\lambda}=\{\lambda_1,\lambda_2,\lambda_3,\lambda_4\}$ to be the dual factor variables, as 
\begin{alignat}{3} \label{eq: the Lagrangian dual function}
\begin{split}
\textnormal{(P3):}& \underset{\substack{\{t_{k,j},t_{k,i},E_{k,j}^{\textnormal{off}},E_{k,i}^{\textnormal{off}}\}}}{\text{min}} \ \mathbb{L}(t_{k,j},t_{k,i},E_{k,j}^{\textnormal{off}},E_{k,i}^{\textnormal{off}}), \\ 
&
 \textnormal{where} \  \mathbb{L}(t_{k,j},t_{k,i},E_{k,j}^{\textnormal{off}},E_{k,i}^{\textnormal{off}})  = E_{k,i}^{\textnormal{off}}+E_{k,j}^{\textnormal{off}} \\ 
& +\lambda_1(L_{k,j} - t_{k,j} R_{k,i}^j) \\ 
&+ \lambda_2 (L_{k,j} - t_{k,i} R_{k,\textnormal{BS}}^{j})\\ 
& + \lambda_3 (L_{k,i} - t_{k,i} R_{k,\textnormal{BS}}^{i}) \\ & + \lambda_4 (t_{k,j}+t_{k,i}-T/K).
\end{split}
\end{alignat}
\indent Thirdly, as (P2) is a convex problem and satisfies Slater's conditions, then (P2) and (P3) have strong duality. By applying the \ac{KKT} conditions and solving the equations~\eqref{eq:partial diff. for t_j} to~\eqref{eq:partial diff. for lambda4}, the optimal solution can be obtained~\cite{boyd2004convex}. 

\begin{strip}

\begin{bottomborder} \end{bottomborder}

\begin{flalign}
\frac{\partial\mathbb{L}}{\partial t_{k,j}}&=\lambda_4 -\lambda_1 W \log_2 (1 + \frac{E_{k,j} |h_{k,ji}|^2}{\sigma^2 t_{k,j}}) +\frac{\lambda_1 W E_{k,j} |h_{k,ji}|^2}{(\sigma^2 t_{k,j}+E_{k,j} |h_{k,ji}|^2)\textnormal{ln2}} =0, && \label{eq:partial diff. for t_j}
\end{flalign}
\begin{flalign}
\frac{\partial\mathbb{L}}{\partial t_{k,i}}&=\lambda_4 - \lambda_2 W \log_2 (1 + \frac{\beta_{k,j}E_{k,i} |\mathbf{h}_{k,i}^{H} \mathbf{w}_{k}|^2}{\sigma^2 t_{k,i}}) + \frac{\lambda_2 W \beta_{k,j} E_{k,i} |\mathbf{h}_{k,i}^{H} \mathbf{w}_{k}|^2}{(\sigma^2 t_{k,i} + \beta_{k,j} E_{k,i} |\mathbf{h}_{k,i}^{H} \mathbf{w}_{k}|^2 )\textnormal{ln2}} && \nonumber \\
& -\lambda_3 W \log_2 (1 + \frac{\beta_{k,i}E_{k,i} |\mathbf{h}_{k,i}^{H} \mathbf{w}_{k} |^2}{\beta_{k,j}E_{k,i} |\mathbf{h}_{k,i}^{H} \mathbf{w}_{k}|^2 + \sigma^2}) && \nonumber \\
& + \frac{\lambda_3 t_{k,i} W \beta_{k,i} E_{k,i} |\mathbf{h}_{k,i}^{H} \mathbf{w}_{k}|^2 \sigma^2}{(\beta_{k,j} E_{k,i} |\mathbf{h}_{k,i}^{H} \mathbf{w}_{k}|^2 + \sigma^2 t_{k,i} + \beta_{k,i} E_{k,i} |\mathbf{h}_{k,i}^{H} \mathbf{w}_{k}|^2)(\beta_{k,j} E_{k,i} |\mathbf{h}_{k,i}^{H} \mathbf{w}_{k}|^2 + \sigma^2 t_{k,i})\textnormal{ln2}} =0, &&\label{eq:partial diff. for t_i}
\end{flalign}
\begin{flalign}
\frac{\partial\mathbb{L}}{\partial E_{k,j}}&=1-\frac{\lambda_1 t_{k,j} W |h_{k,ji}|^2}{(E_{k,j} |h_{k,ji}|^2 + \sigma^2 t_{k,j}) \textnormal{ln2}}=0, &&\label{eq:partial diff. for E_j}
\end{flalign}
\begin{flalign}
\frac{\partial\mathbb{L}}{\partial E_{k,i}}&=1-\frac{\lambda_2 t_{k,i} W \beta_{k,j} |\mathbf{h}_{k,i}^{H} \mathbf{w}_{k}|^2}{(\sigma^2 t_{k,i} + \beta_{k,j} E_{k,i} |\mathbf{h}_{k,i}^{H} \mathbf{w}_{k}|^2)\textnormal{ln2}}&& \nonumber \\
&- \frac{\lambda_3 t_{k,i}^2 W \beta_{k,i} \sigma^2 |\mathbf{h}_{k,i}^{H} \mathbf{w}_{k}|^2 \sigma^2}{(\beta_{k,j} E_{k,i} |\mathbf{h}_{k,i}^{H} \mathbf{w}_{k}|^2 + \sigma^2 t_{k,i}+ \beta_{k,i} E_{k,i} |\mathbf{h}_{k,i}^{H} \mathbf{w}_{k}|^2)(\beta_{k,j} E_{k,i} |\mathbf{h}_{k,i}^{H} \mathbf{w}_{k}|^2 + \sigma^2 t_{k,i})\textnormal{ln2}} =0, &&\label{eq:partial diff. for E_i}
\end{flalign}
\begin{flalign}
\frac{\partial\mathbb{L}}{\partial \lambda_1}&=L_{k,j} - t_{k,j}W \log_2 (1 + \frac{E_{k,j}^{\textnormal{off}} |h_{k,ji}|^2}{\sigma^2 t_{k,j}})=0, &&\label{eq:partial diff. for lambda1}
\end{flalign}
\begin{flalign}
\frac{\partial\mathbb{L}}{\partial \lambda_2}&=L_{k,j} - t_{k,i}W \log_2 (1 + \frac{\beta_{k,j}E_{k,i}^{\textnormal{off}} |\mathbf{h}_{k,i}^{H} \mathbf{w}_{k}|^2}{\sigma^2 t_{k,i}})=0, &&\label{eq:partial diff. for lambda2}
\end{flalign}
\begin{flalign}
\frac{\partial\mathbb{L}}{\partial \lambda_3}&=L_{k,i} - t_{k,i} W \log_2 (1 + \frac{\frac{\beta_{k,i}E_{k,i}^{\textnormal{off}} |\mathbf{h}_{k,i}^{H} \mathbf{w}_{k} |^2}{t_{k,i}}}{\frac{\beta_{k,j}E_{k,i}^{\textnormal{off}} |\mathbf{h}_{k,i}^{H} \mathbf{w}_{k}|^2}{t_{k,i}}+ \sigma^2})=0, &&\label{eq:partial diff. for lambda3}
\end{flalign}
\begin{flalign}
 \frac{\partial\mathbb{L}}{\partial \lambda_4}&=t_{k,j}+t_{k,i}-T/K=0.&&\label{eq:partial diff. for lambda4}
\end{flalign}

\begin{bottomborder} \end{bottomborder}

\end{strip}

Specifically, the optimal solution of $t_{k,i}^*$ can be found by approximating~\eqref{eq:partial diff. for lambda3} in the high \ac{SNR} region, $\frac{\beta_{k,i}E_{k,i}^{\textnormal{off}} |\mathbf{h}_{k,i}^{H} \mathbf{w}_{k} |^2}{t_{k,i}} >> \sigma^2$, and solving the following approximated expression
\begin{equation}
\frac{\partial\mathbb{L}}{\partial \lambda_3} \approx L_{k,i} - t_{k,i} W \log_2 (1 + \frac{\beta_{k,i}}{\beta_{k,j}})=0, \label{eq: approx. partial diff. for lambda3}
\end{equation}
\noindent after getting $t_{k,i}^*$, the remaining optimal values for $t_{k,j}^*,E_{k,j}^*,E_{k,i}^*$ can be obtained by solving~\eqref{eq:partial diff. for lambda4}, \eqref{eq:partial diff. for lambda1}, and~\eqref{eq:partial diff. for lambda2}, respectively. Subsequently, one can get $p_{k,j}^*,p_{k,i}^*$ by substituting $p_{k,j}^*=E_{k,j}^*/t_{k,j}^*$ and $p_{k,i}^*=E_{k,i}^*/t_{k,i}^*$. With this, the proof is
completed. 

\bibliographystyle{IEEEtran}
\bibliography{main}


\begin{IEEEbiography}[{\includegraphics[width=1in,height=1.25in,clip,keepaspectratio]{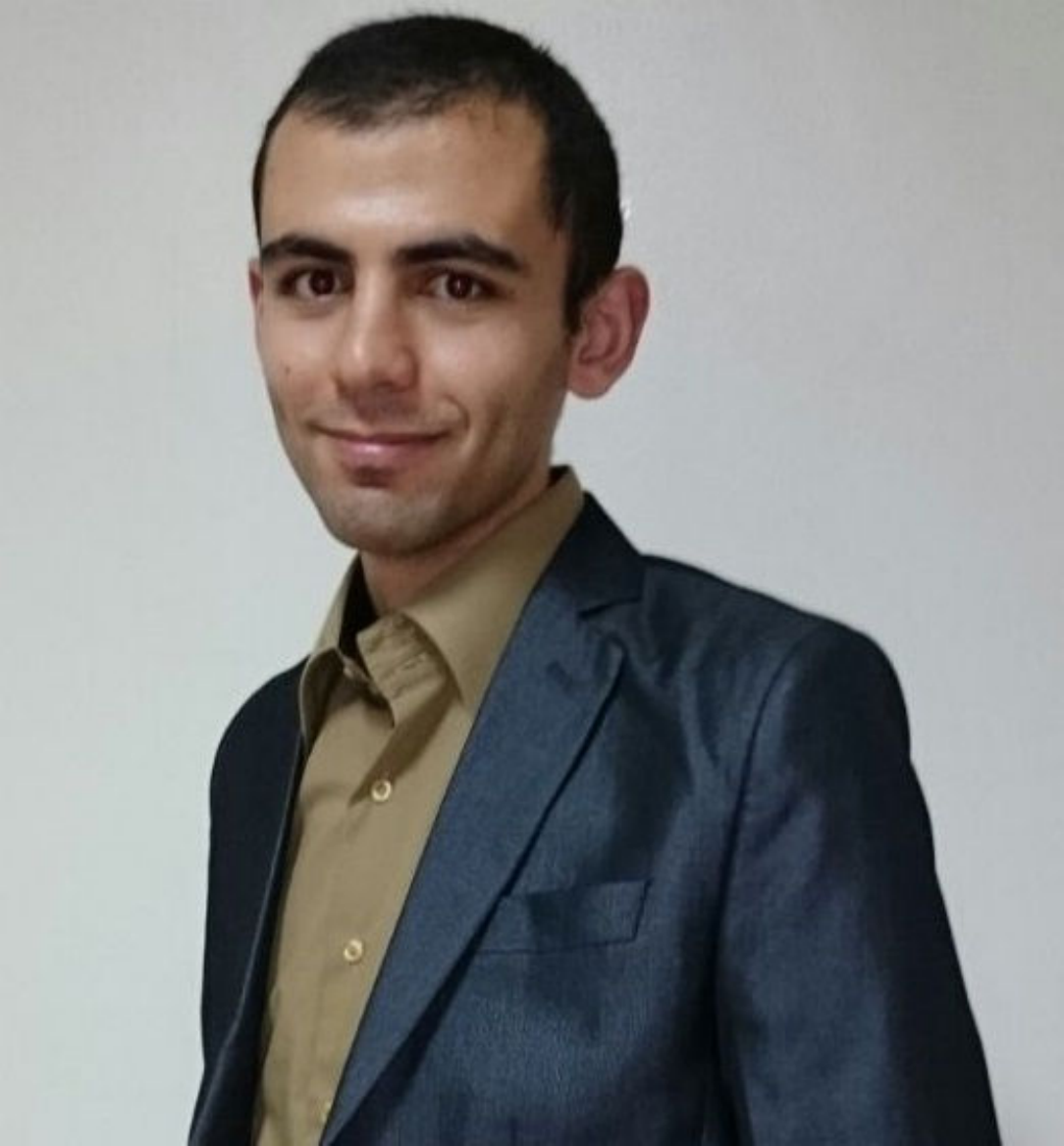}}]{\textbf{Omar Maraqa}}
has received his B.S. degree in Electrical Engineering from Palestine Polytechnic University, Palestine, in 2011, his M.S. degree in Computer Engineering from King Fahd University of Petroleum \& Minerals (KFUPM), Dhahran, Saudi Arabia, in 2016, and his Ph.D. degree in Electrical Engineering at KFUPM, Dhahran, Saudi Arabia, in 2022. He is currently a Postdoctoral Research Fellow with the Department of Electrical and Computer Engineering, at McMaster University, Canada. His research interests include performance analysis and optimization of wireless communications systems.
\end{IEEEbiography}


\begin{IEEEbiography}[{\includegraphics[width=1in,height=1.25in,clip,keepaspectratio]{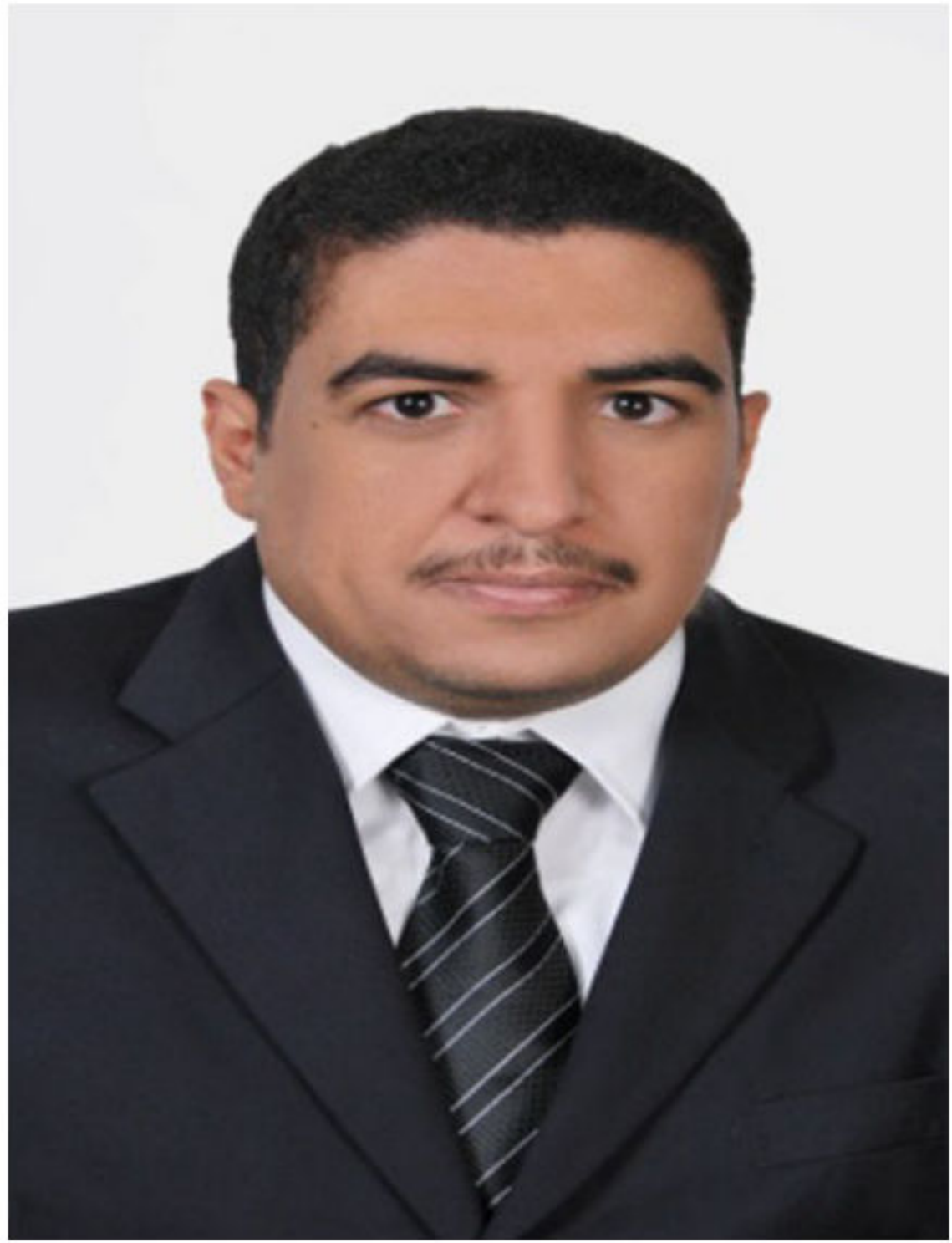}}]{\textbf{Saad Al-Ahmadi}}
has received his M.Sc. in Electrical Engineering from King Fahd University of Petroleum \& Minerals (KFUPM), Dhahran, Saudi Arabia, in 2002 and his Ph.D. in Electrical and Computer Engineering from Ottawa-Carleton Institute for ECE (OCIECE), Ottawa, Canada, in 2010. He is currently with the Department of Electrical Engineering at KFUPM as an Associate Professor. His current research interests include channel characterization, design, and performance analysis of wireless communications systems and networks.
\end{IEEEbiography}


\begin{IEEEbiography}[{\includegraphics[width=1in,height=1.25in,clip,keepaspectratio]{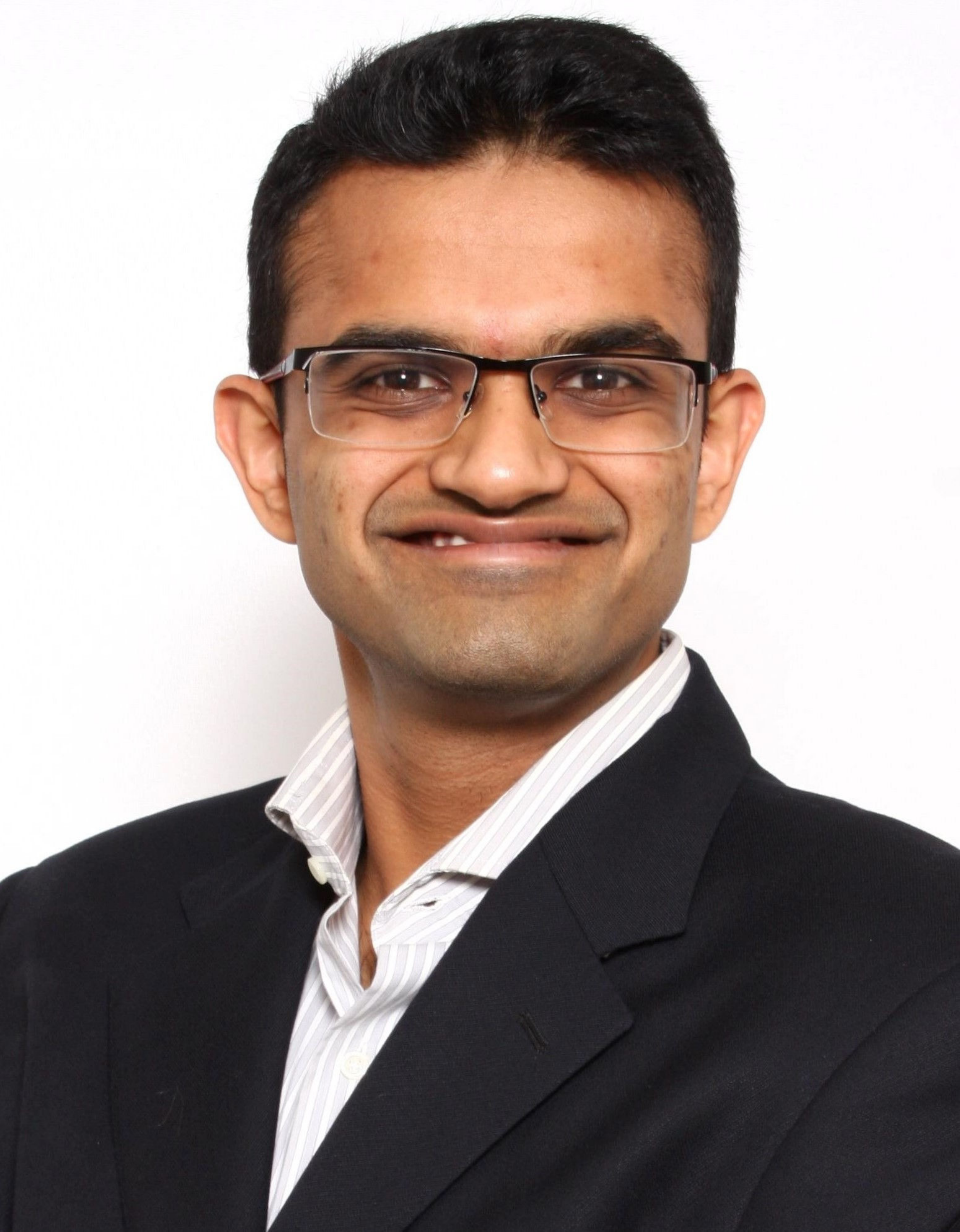}}]{Aditya S. Rajasekaran} (M'18) received the B.Eng (with High Distinction) and M.Eng degrees in Systems and Computer Engineering from Carleton University, Ottawa, ON, Canada, in 2014 and 2017, respectively. He is currently pursuing his Ph.D. degree, also in Systems and Computer Engineering at Carleton University. His research interests include wireless technology solutions for 5G and beyond cellular networks, including non-orthogonal multiple access solutions. He is also with Ericsson Canada, where he has been working as a software developer since 2014. He is currently involved in the physical layer development work for Ericsson's 5G New Radio (NR) solutions.
\end{IEEEbiography}


\begin{IEEEbiography}[{\includegraphics[width=1in,height=1.25in,clip,keepaspectratio]{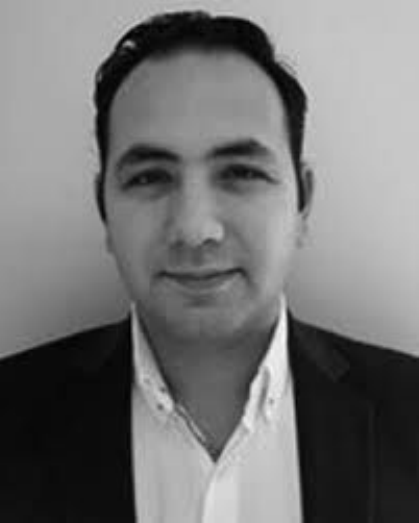}}]{\textbf{Hamza U. Sokun}}
received the B.Sc. degree in electronics engineering from Kadir Has University, Istanbul, Turkey, in 2010, the M.Sc. degree in electrical engineering from Ozyegin University, Istanbul, Turkey, in 2012, and the Ph.D. degree in electrical engineering from Carleton University, Ottawa, Canada, in 2017. He is currently working as a 5G system developer at Ericsson Canada, Ottawa, Canada. His research interests include signal processing and wireless communications.
\end{IEEEbiography}


\begin{IEEEbiography}[{\includegraphics[width=1in,height=1.25in,clip,keepaspectratio]{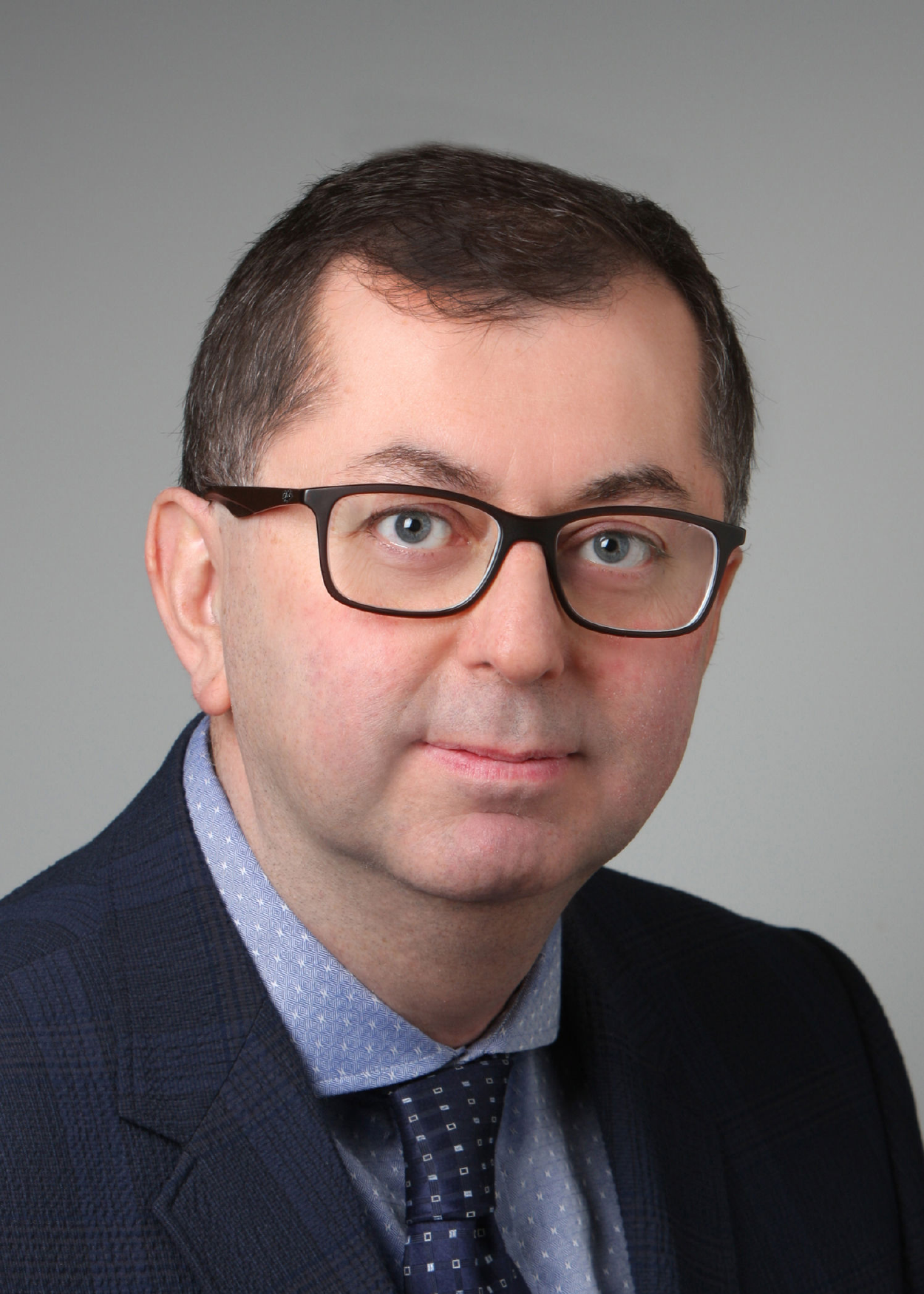}}]{Halim Yanikomeroglu} (F'17) is a full professor in the Department of Systems and Computer Engineering at Carleton University, Ottawa, Canada. His research covers many aspects of communications technologies with an emphasis on wireless networks. His collaborative research with the industry has resulted in 39 granted patents (plus more than a dozen applied). He is a Fellow of the Engineering Institute of Canada (EIC) and the Canadian Academy of Engineering, and he is a Distinguished Lecturer for IEEE Communications Society and IEEE Vehicular Technology Society. He has also supervised 26 Ph.D. students (all completed with theses).
\end{IEEEbiography}


\begin{IEEEbiography}[{\includegraphics[width=1in,height=1.25in,clip,keepaspectratio]{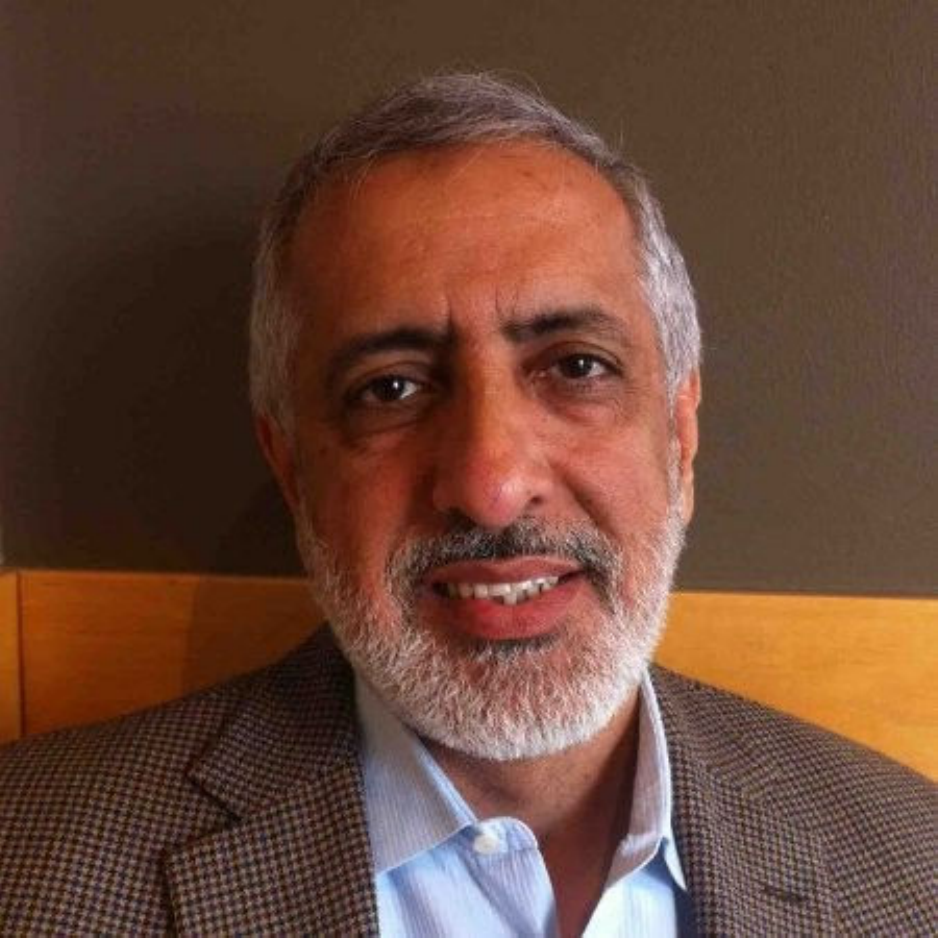}}]{\textbf{Sadiq M. Sait}}
(SM'02) received a bachelor's degree in electronics engineering from Bangalore University in 1981, and a master's and Ph.D. degrees in electrical engineering from the King Fahd University of Petroleum \& Minerals (KFUPM) in 1983 and 1987, respectively. He is currently a Professor of Computer Engineering and the Director of the Center for Communications and IT Research, Research Institute, KFUPM. He has authored over 200 research papers, contributed chapters to technical books, and lectured in over 25 countries. He is also the Principle Author of two books. He received the Best Electronic Engineer Award from the Indian Institute of Electrical Engineers, Bengaluru, in 1981.
\end{IEEEbiography}


\end{document}